\documentclass[11pt]{article}
\usepackage{style}
\def\COUNT{\text{\sf COUNT}}

\begin{document}

\title{The complexity of counting planar graph homomorphisms of domain size 3}
\author{
  Jin-Yi Cai\thanks{Department of Computer Sciences, University of Wisconsin-Madison.
  Partial support for this research was provided by the Office of VCRGE at UW-Madison
  with funding from WARF, and a Simons Fellowship.}\\
  \texttt{jyc@cs.wisc.edu}
  \and
  Ashwin Maran\footnotemark[1]\\
  \texttt{amaran@wisc.edu}
}
\date{}
\maketitle

\begin{abstract}
    We prove a complexity dichotomy theorem for counting planar graph homomorphisms of domain size 3.
    Given any 3 by 3 real valued symmetric matrix $H$ defining a
    graph homomorphism from all planar graphs $G \mapsto Z_H(G)$, we
    completely classify
    the computational complexity of this problem according to
    the matrix $H$. We show that for every $H$, the problem
    is either polynomial time computable or \#P-hard.
    The P-time computable cases consist of
    precisely those that are  P-time computable
    for general graphs (a complete classification is known~\cite{goldberg2010complexity}) or computable by
    Valiant's holographic algorithm via matchgates.
    We also prove several results about planar graph homomorphisms  for general domain size $q$.
    The proof uses mainly analytic arguments.
\end{abstract}

\section{Introduction}\label{sec: full_introduction}

Given graphs $G$ and $H$, a mapping from $V(G)$ to $V(H)$is called
a \textit{homomorphism} if edges of $G$ are mapped to edges of $H$.
This is put in a more general or quantitative setting by the notion of a \textit{partition function}.
Let  $M= (m_{i, j})$ be a symmetric $q \times q$ matrix.
In this paper we consider arbitrary real entries $m_{i, j} \in \mathbb{R}$; 
if $m_{i, j} \in \{0, 1\}$ (or $m_{i, j} \ge 0$), then $M$ is 
the unweighted (or nonnegatively weighted) adjacency matrix of a graph
$H = H_M$. Given $M$,
the partition function  $Z_{M}(G)$ for any input undirected multi-graph $G = (V, E)$
is defined as
$$Z_{M}(G) = \sum_{\sigma: V \rightarrow [q]} \prod_{(u, v) \in E} m_{\sigma(u), \sigma(v)}.$$
Obviously isomorphic graphs $G \cong G'$ have the same value $Z_{M}(G) = Z_{M}(G')$, and thus every $M$
defines a graph property $Z_{M}(\cdot)$. For a 0-1 matrix $M$,
$Z_{M}(G)$ counts the number of homomorphisms from $G$ to $H$.
Graph homomorphism (GH) encompasses a great deal of graph properties~\cite{lovasz2012large}.

Each $M$ defines a
computational problem, denoted by $\EVAL(M)$: given an input graph $G$ compute $Z_{M}(G)$.
The  complexity of 
$\EVAL(M)$ has been a major focus of research.
A number of increasingly general complexity dichotomy theorems have been achieved.
Dyer and Greenhill \cite{dyer2000complexity} proved that for any 0-1 symmetric matrix $M$, computing $Z_{M}(G)$ is either in P-time or is $\#$P-complete. Bulatov and Grohe \cite{bulatov2005complexity}  found a complete classification for $\EVAL(M)$ for all nonnegative matrices $M$. Goldberg, Grohe, Jerrum and Thurley \cite{goldberg2010complexity} then proved a dichotomy for all real-valued matrices $M$. Finally, Cai, Chen, and Lu \cite{cai2013graph} established a dichotomy for all complex valued matrices $M$.
We also note that graph homomorphism can be viewed as a special case of
counting CSP, with one binary constraint function. For counting CSP,
a series of results established a complexity dichotomy for any set of constraint functions
${\cal F}$, going from 
0-1 valued~\cite{bulatov2013complexity,dyer2010complexity,dyer2011decidability,dyer2013complexity} to nonnegative rational valued~\cite{bulatov2012complexity}, to
nonnegative real valued~\cite{cai2016nonnegative}, to all complex valued functions~\footnote{The
 last counting CSP dichotomy~\cite{cai2017complexity} is not known to have a decidable dichotomy criterion, 
while the dichotomy criterion on GH~\cite{cai2013graph} is polynomial-time decidable. So, \cite{cai2017complexity} does not strictly supersede \cite{cai2013graph}.}.

Parallel to this development, Valiant~\cite{valiant2008holographic} introduced
\emph{holographic algorithms}. It is well known that counting the number of perfect
matchings (\#PM) is \#P-complete~\cite{valiant1979complexity}. On the other hand,
since the 60's, there has been a famous FKT algorithm~\cite{kasteleyn1961statistics,temperley1961dimer,kasteleyn1963dimer,kasteleyn1967graph} that can
compute \#PM on planar graphs in P-time. Valiant's holographic algorithms
greatly extended its reach, in fact so much so that a most intriguing question arises: Is this   a \emph{universal} algorithm that 
every counting problem (expressed as a sum-of-products)
that \emph{can be solved} in P-time on planar graphs
(but \#P-hard in general) \emph{is solved} by this method alone?
Such a universality statement must appear to be extraordinarily, if not overly, ambitious.

To attack this problem, Holant problems are introduced~\cite{cai2009holant}.
Holant problems are edge-models,
where for an input graph, constraint functions
are attached to vertices, and the edges serve as variables.
Typical examples are \#PM, counting proper edge colorings, or cycle covers, etc.
It can be shown that counting CSP can be 
expressed as Holant problems, but not conversely~\cite{freedman2007reflection}.
It is in the framework of Holant problems, a classification of
such counting problems can be studied, and the power of holographic algorithms
be understood.

After a series of work~\cite{cai2009holant,cai2016complete,cai2019holographic, backens2017new, backens2018complete, yang2022local, fu2019blockwise, fu2014holographic, cai2019holographic}
it was established that for every set of complex valued constraint functions ${\cal F}$
on the Boolean domain (i.e., $q=2$) there is a 3-way
classification for  
\#CSP(${\cal F}$): 
 (1) P-time solvable, 
(2) P-time solvable over planar graphs but \#P-hard over general graphs, 
(3) \#P-hard over planar graphs.
Moreover,
category (2) consists of precisely  those problems that can be solved by
Valiant's holographic algorithm using FKT.
Note that, curiously, 
this is a reduction to \#PM, which is a Holant problem, but not a \#CSP problem.
More mysteriously, it is further proved that for the broader class of
Holant problems on the Boolean domain, Valiant's holographic algorithm is 
\emph{not} universal for category (2)~\cite{cai2015holant}.
So far we have very limited knowledge on this universality question
for higher domain problems ($q >2$). Before this work, no complexity
classification was known for the \textit{planar} version of  $\EVAL(M)$ for $q=3$, even for 0-1 matrices $M$.

On the other hand, the  \textit{planar} version of $\EVAL$
(for general $q$) has been found to be intimately related
to quantum information theory~\cite{banica2005quantum,cleve2017perfect,atserias2019quantum,musto2019morita,mancinska2014graph,lupini2020nonlocal}.
 Man{\v{c}}inska and Roberson~\cite{manvcinska2020quantum}
showed that
two graphs $H$ and $H'$ are \emph{quantum isomorphic} iff for every
\emph{planar} input graph $G$, the partition function $Z_H(G) = Z_{H'}(G)$.
This is in contrast to a classical result by Lov\'asz~\cite{lovasz1967operations}
that $H$ and $H'$ are \emph{isomorphic} iff  $Z_H(G) = Z_{H'}(G)$ for \emph{every} graph $G$.
They further proved that it is in general undecidable for two graphs $H$ and $H'$,
whether  $Z_H(G) = Z_{H'}(G)$ for all \emph{planar}  graphs $G$~\cite{manvcinska2020quantum,atserias2019quantum}. ($H$ and $H'$ need not be planar.)

Our goal is strictly on the complexity question. 
What is the computational complexity of $Z_M(G)$ from  planar input graphs $G$?
This paper marks the beginning of this quest.
Let $\PlEVAL(M)$ denote the problem $\EVAL(M)$ 
when the input graphs are restricted to planar graphs.
(Again, the underlying graph $H_M$ is not restricted to planar graphs.)
For domain size $q=3$, we  give a complete classification of the complexity of
$\PlEVAL(M)$ for all real valued matrices $M$.
We prove that an exact classification according to the three categories above hold
for this class,
and a holographic reduction to FKT remains 
a \emph{universal} algorithm for category (2).
As in previous work, generalizing a dichotomy 
from domain size 2 to domain size 3 
has to overcome significant difficulty 
and can lead to future progress~\cite{bulatov2002dichotomy}.
We also prove several results about  $\PlEVAL(M)$  for general $q$.
For example, we give a generic criterion that leads to \#P-hardness of
$\PlEVAL(M)$ for non-negative real matrices (\cref{theorem: full_qTimesqNonBipLatticeNonNegative,theorem: full_qTimesqBipartitieLatticeNonNegative}),
and also prove that  $\PlEVAL(M)$ is $\#$P-hard for almost all
 $M$ (\cref{theorem: full_mesaure0}).

Now we give some highlights of the proof  for the $q=3$ case.
First, we use the Boolean domain dichotomy to handle certain
$3 \times 3$ matrices. This includes both the  \emph{reducible} matrices
as well as a subtler case called \emph{twinned} matrices.
For the latter, we can transform the problem
to a version of the partition function
with degree dependent vertex weights. Then
we use a gadget construction from~\cite{govorov2020dichotomy} 
and interpolation to get rid of
the dependency on vertex degree.

We then formulate a \emph{lattice condition} on the eigenvalues
of $M$, which if satisfied, would allow us to carry out a successful
\#P-hardness proof using Vandermonde systems. After some work,
%
it boils down to proving that  $\PlEVAL(M(p))$  is \#P-hard 
for some real $p>1$, where $M(p)$ is a family of
matrices of the form $(p^{x_{ij}})$
for some integers $x_{ij} \ge 0$.  We have very little control
of $x_{ij}$ except that $M(p)$ has full rank in some small interval
$p \in I_{\epsilon} = (1, 1+\epsilon)$ for some $\epsilon > 0$. Let $\lambda_i = \lambda_i(p)$ be the eigenvalues of $M(p)$
ordered by $|\lambda_1|  \le |\lambda_2| \le |\lambda_3|$.
By the Perron theorem we have $|\lambda_2| < |\lambda_3|$
is strict and there is a well defined and unique  $t(p) \in (0, 1]$
such that $|\lambda_2| = |\lambda_1|^{t(p)} |\lambda_3|^{1-t(p)}$.

The crux of the proof is to show the following: ({\sc i})
$\lim_{p \rightarrow 1^+}t(p)$ exists, and is either 1/2 or 1.
({\sc ii}) $t(p)$ cannot be identically 1/2 in $I_{\epsilon}$.
({\sc iii}) If $t(p)$ is identically 1 in $I_{\epsilon}$, then  $\PlEVAL(M)$  is \#P-hard.
From ({\sc i}), if $t(p)$ is constant in $I_{\epsilon}$ it can only be   1/2 or 1.
From ({\sc ii}) and ({\sc iii}), we may assume neither case holds,
and so  $t(p)$ is not constant.
Thus, by the intermediate value theorem there exists
some $p \in I_{\epsilon}$ such that $t(p)$ is \emph{irrational}. 
This irrational  $t(p)$ will fulfill the lattice condition!

Now a perceptive reader may object that the $p$ that
produces an irrational $t(p)$ may not even be algebraic, and
the usual definition for the complexity
of partition functions requires that the real numbers be
algebraic so that strict bit complexity in terms of Turing machines can apply.

This is a serious quandary. Our proof is intrinsically  analytic.
Also
there are indeed non-constant continuous functions $f(\cdot)$  that map
all algebraic $p$ to rational 
$f(p)$ (see \cref{sec: full_appendixB}). 
Furthermore, it seems hopeless to prove that our $t(p)$ 
is not
 such a function (although probably true). 

We resolve this difficulty by a bold approach---we will allow 
\emph{all} real $M$ for $\EVAL(M)$ and still stay within strict bit complexity  of Turing machines.
This uses the theorem of unique transcendence degree~\cite{jacobson1985basic}.
Details are in \cref{sec: full_modelComputation}.
We remark that this makes it 
non-constructive.
For instance, it is unknown whether $\mathbb{Q}(e, \pi)$
has transcendence degree 1 or 2. If $M$ has  both $e$
and $\pi$ our formal definition of $\EVAL(M)$ treats this degree
as ``known'' (existentially); $M$ is a fixed constant
for the computational problem $\EVAL(M)$, and the complexity statements refer to the
\emph{existence} of either P-time algorithms or reductions
(but not how to get them).

\section{Preliminaries}\label{sec: full_preliminaries}

Let $M = (m_{i, j})$ be a symmetric $q \times q$ real valued matrix, $i, j \in [q]$. Given a planar, undirected multi-graph $G = (V, E)$, we can perform certain elementary operations (that preserve planarity) on the graph $G$ to transform it into a new graph $G'$, such that $Z_{M}(G') = Z_{M'}(G)$ for some matrix $M'$. For most of  this paper we will use  two such operations, thickening and stretching.

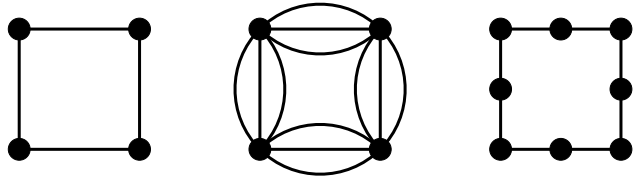
\begin{figure}[b]
	\centering
	\scalebox{0.8}{\begin{tikzpicture}[line join=miter, draw opacity=1]

\node[circle, draw=black, fill=black] (c) at (-1, -1){};
\node[circle, draw=black, fill=black] (c) at (-1, 1){};
\node[circle, draw=black, fill=black] (c) at (1, 1){};
\node[circle, draw=black, fill=black] (c) at (1, -1){};

\draw[line width=0.5mm, black] (-1, -1) -- (-1, 1);
\draw[line width=0.5mm, black] (-1, 1) -- (1, 1);
\draw[line width=0.5mm, black] (1, 1) -- (1, -1);
\draw[line width=0.5mm, black] (1, -1) -- (-1, -1);

\begin{scope}[xshift = 4cm]

\node[circle, draw=black, fill=black] (c) at (-1, -1){};
\node[circle, draw=black, fill=black] (c) at (-1, 1){};
\node[circle, draw=black, fill=black] (c) at (1, 1){};
\node[circle, draw=black, fill=black] (c) at (1, -1){};

\draw[line width=0.5mm, black] (-1, -1) -- (-1, 1);
\draw[line width=0.5mm, black] (-1, 1) -- (1, 1);
\draw[line width=0.5mm, black] (1, 1) -- (1, -1);
\draw[line width=0.5mm, black] (1, -1) -- (-1, -1);

\draw[line width=0.5mm, black] (-1, -1) to [in=225, out=135] (-1, 1);
\draw[line width=0.5mm, black] (-1, -1) to [in=-45, out=45] (-1, 1);

\draw[line width=0.5mm, black] (-1, 1) to [in=135, out=45] (1, 1);
\draw[line width=0.5mm, black] (-1, 1) to [in=225, out=-45] (1, 1);

\draw[line width=0.5mm, black] (1, 1) to [in=45, out=-45] (1, -1);
\draw[line width=0.5mm, black] (1, 1) to [in=135, out=225] (1, -1);

\draw[line width=0.5mm, black] (1, -1) to [in=45, out=135] (-1, -1);
\draw[line width=0.5mm, black] (1, -1) to [in=-45, out=225] (-1, -1);

\end{scope}

\begin{scope}[xshift = 8cm]

\node[circle, draw=black, fill=black] (c) at (-1, -1){};
\node[circle, draw=black, fill=black] (c) at (-1, 1){};
\node[circle, draw=black, fill=black] (c) at (1, 1){};
\node[circle, draw=black, fill=black] (c) at (1, -1){};

\draw[line width=0.5mm, black] (-1, -1) -- (-1, 1);
\draw[line width=0.5mm, black] (-1, 1) -- (1, 1);
\draw[line width=0.5mm, black] (1, 1) -- (1, -1);
\draw[line width=0.5mm, black] (1, -1) -- (-1, -1);

\node[circle, draw=black, fill=black] (c) at (-1, 0){};
\node[circle, draw=black, fill=black] (c) at (0, 1){};
\node[circle, draw=black, fill=black] (c) at (1, 0){};
\node[circle, draw=black, fill=black] (c) at (0, -1){};

\end{scope}

\end{tikzpicture}}
	\caption{A graph $G$, the thickened graph $T_{3}G$, and the stretched graph $S_{2}G$.}
	\label{fig:thickeningStretching}
\end{figure}

From any planar multi-graph $G = (V, E)$, and a positive integer $k$, we can construct the planar multi-graph $T_{k}G$, by replacing every edge in $G$ with $k$ parallel edges between the same vertices. This process is called \emph{thickening}.
Clearly $Z_{M}(T_{k}G) = Z_{T_{k}M}(G)$, where $T_{k}M$ is the $q \times q$ matrix with entries $\big((m_{i, j})^{k}\big)$ for $i, j \in [q]$. In particular, $\PlEVAL(T_{k}M) \leq \PlEVAL(M)$ for all $k \geq 1$.

Similarly, from any planar multi-graph $G = (V, E)$, and a positive integer $k$, we can construct the planar multi-graph $S_{k}G$ by replacing every edge $e \in E$ with a path of length $k$. This process is called \emph{stretching}.
It is also easily seen that  $Z_{M}(S_{k}G) = Z_{S_{k}M}(G)$, where $S_{k}M = M^{k}$, the $k$-th power of $M$. So, we also have  $\PlEVAL(S_{k}M) \leq \PlEVAL(M)$ for all $k \geq 1$.

\subsection{Model of Computation}\label{sec: full_modelComputation}

The Turing machine model is naturally suited to the study of computation over discrete structures such as integers or graphs. 
When  $M \in \mathbb{R}^{q \times q}$, for  $\PlEVAL(M)$ 
one usually restricts $M$ to be a matrix
 with only algebraic numbers. This is strictly for the consideration
 of the model of computation, even though allowing all
 real-valued matrices would be more natural. 
 
 There is a formal (albeit nonconstructive) method to treat  $\PlEVAL(M)$ 
 for arbitrary real-valued matrices $M$ and yet stay strictly
 within the Turing machine model in terms of  bit-complexity.
In this paper, because our proof depends heavily on analytic
argument with continuous functions on the real line, this
logical formal view becomes necessary.

To begin with, we recall a theorem from field theory:
Every extension field ${\bf F}$ over  ${\mathbb Q}$
by a finite set of real numbers is  a finite algebraic extension ${\bf E}'$
of a certain  purely transcendental   extension field ${\bf E}$
over ${\mathbb Q}$,
which has the form ${\bf E} = {\mathbb Q}(X_1,
\ldots, X_m)$ where $m \ge 0$ and $X_1,
\ldots, X_m$ are algebraically independent~\cite{jacobson1985basic} (Theorem 8.35, p.~512).
${\bf F}$ is said to have
a finite transcendence degree $m$ over ${\mathbb Q}$.
It is known that $m$ is uniquely defined for ${\bf F}$.
Since 
$\rm{char}~{\mathbb Q} =0$,
the finite algebraic extension ${\bf E}'$ over ${\bf E}$
is actually simple, ${\bf E}' = {\bf E}(\beta)$ for some $\beta$,
and it is specified by a
minimal polynomial in ${\bf E}[X]$.
Now given a real matrix $M$, let ${\bf F} = {\mathbb Q}(M)$ 
be the extension field by adjoining the entries of $M$.
We  consider $M$ is fixed for the problem $\PlEVAL(M)$,
and thus we may assume (nonconstructively) that the form
${\bf F} = {\bf E}(\beta)$
 and ${\bf E} = {\mathbb Q}(X_1,
\ldots, X_m)$ are given. (This means, among other things,
that the minimal polynomial of $\beta$ over ${\bf E}$ is given,
and all arithmetic operations can be performed on ${\bf F}$.)

Now, the computational problem $\PlEVAL(M)$ is the following:
Given a planar 
$G$, compute $Z_{M}(G)$ as an element in ${\bf F}$
(which is expressed as a polynomial in $\beta$ with coefficients in ${\bf E}$).
More concretely, we can show that this is equivalent to the following problem $\COUNT(M)$:
The input  is a pair $(G,x)$,
  where $G=(V,E)$ is a planar graph and $x\in {\bf F}$.
The output is\vspace{-0.1cm}
$$
\text{\#}_{M}(G,x)= \Big|\big\{\sigma:V\rightarrow
  [q]\hspace{0.08cm}: \hspace{0.08cm} \prod_{(u, v) \in E} m_{\sigma(u), \sigma(v)}=x\big\}\Big|,\label{full_COUNTM}
  $$
a non-negative integer. Note that, in this definition,
we are basically combining terms with the same 
product value in the definition of $Z_{M}(G)$.

Let $n=|E|$.
Define  $X$ to be the  set of all possible product values
appearing in $Z_{M}(G)$:
\begin{equation}\label{full_definitionreuse}
X=\left\{\prod_{i,j\in [q]}m_{ij}^{k_{ij}}\hspace{0.08cm}\Big|\hspace{0.1cm}
  \text{integers $k_{ij}\ge 0$ and $\sum_{i,j\in [q]}k_{ij}=n$}
\right\}.
\end{equation}
There are $\binom{n+q^2-1}{q^2-1} = n^{O(1)}$ many integer sequences
$(k_{i,j})$ such that  $k_{i,j}\ge 0$ and $\sum_{i,j\in [q]}k_{i,j}=n$.
$X$
  is defined as a set, not a multi-set.
After removing repeated
elements the cardinality $|X|$ is also polynomial in $n$.
 For fixed and given ${\bf F}$
 the elements in $X$ can be enumerated in polynomial time in $n$.
 (It is important that ${\bf F}$ and $q$ are all treated as fixed constants.)
It then follows from the definition that
  $\text{\#}_{M}(G,x)= 0$ for any $x\notin X$.
This gives us the following relation:
\[Z_M(G)= \sum_{x\in X} x \cdot \text{\#}_{M}(G,x),\ \ \ \text{for any 
  graph $G$,}\]
and thus, $\PlEVAL(M)\le \COUNT(M).$

For the other direction,
  we construct, for any $p\in [|X|]$ (recall that $|X|$ is polynomial in $n$),
  a  planar graph $T_pG$ from $G$ by replacing
  every edge of $G$ with $p$ parallel edges.
Then,
$$
Z_M(T_pG)= \sum_{x\in X} x^p \cdot \text{\#}_{M}(G,x),\ \ \ \text{for any 
  graph $G$.}
$$
This is a Vandermonde system; it has full rank since
elements in $X$ are distinct by definition. So by
querying $\PlEVAL(M)$ for the
  values of $Z_{M}(T_pG)$,
  we can solve it in polynomial time
  and get $\text{\#}_{M}(G,x)$ for\vspace{0.0015cm} every non-zero $x\in X$.
To obtain $\text{\#}_{M}(G,0)$ (if $0\in X$), we note that
$$
\sum_{x\in X} \text{\#}_{M}(G,x) =q^{|V|}.
$$
This gives us a polynomial-time reduction and thus, $\COUNT(M)\le \PlEVAL(M)$.
%
We have proved 
\begin{lemma}\label{lemma: full_count}
For any fixed  $M \in \mathbb{R}^{q \times q}$,  
   $\PlEVAL(M)\equiv \COUNT(M)$.
\end{lemma}
Thus, $\PlEVAL(M)$ can be identified with the 
problem of producing those polynomially many integer coefficients
in the canonical expression for $Z_M(G)$ as a sum
of (distinct) terms from $X$.

This  formalistic view has 
the advantage that we can treat the complexity 
of  $\PlEVAL(M)$ for  general $M$, and not restricted to algebraic numbers. 
 Thus, numbers such as $e$ or $\pi$
need not be excluded.
More importantly, in this paper this generality is essential, due to the proof technique that we employ.
Furthermore, once freed from this restriction we in fact explicitly use
transcendental numbers as a tool in our proof (see \cref{lemma: full_latticeHardnessDelta}).
In short, in this paper, treating the complexity of  $\PlEVAL(M)$ for  general real $M$
is not a \emph{bug} but a \emph{feature}.

 However, we note that this treatment 
 has the following subtlety.  
For the computational problem $\PlEVAL(M)$
the formalistic view demands that
${\bf F}$ be specified in the form ${\bf F} = {\bf E}(\beta)$.
Such a form exists, and its specification is of
constant size when  measured
in terms of the size  of the input graph $G$. 
However, 
in reality many basic questions for transcendental numbers
are unknown.
For example, it is still unknown whether $e + \pi$ or $e \pi$ are
rational, algebraic irrational or transcendental,
and it is open whether ${\mathbb Q}(e, \pi)$ has  transcendence degree
2 (or 1) over  ${\mathbb Q}$, i.e., whether $e$ and $\pi$ are algebraically
independent.
The formalistic view here non-constructively
assumes this information is given for ${\bf F}$.
A  polynomial time reduction $\Pi_1  \le \Pi_2$  from one problem
to another  in this setting merely
implies that the \emph{existence} of a polynomial time algorithm
for $\Pi_2$ logically implies the 
\emph{existence} of  a  polynomial time algorithm
for $\Pi_1$. We do not actually obtain such
an algorithm constructively. 

This logical detour not withstanding, if a reader
is  interested only in the complexity of $\PlEVAL(M)$ 
for integer matrices $M$, then the
complexity dichotomy proved in this paper holds
according to the standard definition of $\PlEVAL(M)$ 
for integral $M$ in terms of
the model of computation; the fact that this is proved 
in a broader setting for all real matrices $M$ 
is irrelevant. This  is akin to the situation
in analytic number theory, where one might be 
interested in a question strictly about the ordinary
integers, but the theorems are proved
in a broader setting of analysis.

\section{Reduction to Boolean domain matrices}\label{sec: full_booleanDomain}

In this section, we handle those $3 \times 3$ matrices for which the planar graph homomorphism problem is equivalent to the same
problem on a $2 \times 2$ matrix. We first state the
following theorem~\cite{guo2020complexity}:

\begin{theorem}\label{theorem: full_domain2Hardness}
	The problem $\PlEVAL(M)$ 	is $\#$P-hard, for
	$M = \left( \begin{smallmatrix}
	x & y\\
	y & z\\
	\end{smallmatrix} \right)$,
 unless one of the following conditions holds, in which case $\PlEVAL(M)$ is polynomially tractable:
 \[   \mbox{(1)~~} xz = y^2,~~~~
	\mbox{(2)~~} y = 0,~~~~
	\mbox{(3)~~} x = z, ~~~~\mbox{or}~~~~
	\mbox{(4)~~} xz = -y^2 ~~\&~~ x = -z. \]
\end{theorem}

Case $\mbox{(3)}$ (the Ising model) is
precisely the problems for which $\PlEVAL(M)$ is P-time solvable,
but $\EVAL(M)$ is $\#$P-hard; these are also
exactly the ones solvable by a holographic algorithm using matchgates.

\subsection{Reducible matrices}\label{sec: full_domainSeparable}

\begin{definition}\label{definition: full_domainSeparableMatrices}
    A $q \times q$ symmetric matrix $M$ is \emph{reducible}
    if there exists a permutation matrix $P$, such that
    $P M P^{\tt T}$ is a direct sum $A \oplus B 
    =\begin{pmatrix} A & \mathbf{0} \\ \mathbf{0} & B \end{pmatrix}$
    for some (nonempty) matrices $A$ and $B$. A symmetric matrix $M$ is \emph{irreducible} if
    it is not reducible.
    Specialized to $q=3$, 
    a symmetric matrix $M$ is reducible
    if it has the form 
    \begin{equation}\label{eqn: full_disconnectedMatrix}
        M^\pi = \begin{pmatrix}
    		x & y & 0\\
    		y & z & 0\\
    		0 & 0 & w\\
    	\end{pmatrix}
	\end{equation}
	after its rows and columns are permuted by a permutation $\pi$.
\end{definition}

As $M$ is symmetric, the permutation $\pi$ on the
rows and columns must be the same. Clearly $Z_M(G) =
Z_{M^\pi}(G)$ for all $G$.
For a reducible $q \times q$ matrix $M^{\pi} = A \oplus B$, 
it is known from \cite{cai2013graph} (Lemma 4.3) 
that $\EVAL(M)$ is $\#$P-hard iff
at least one of  $\EVAL(A)$ or $\EVAL(B)$ is $\#$P-hard,
and $\EVAL(M)$ is P-time tractable iff both  $\EVAL(A)$ or $\EVAL(B)$ are P-time tractable.
The P-time tractablility statement holds for $\PlEVAL$ as well.
Also, it can be checked that the proof 
for $\#$P-hardness  in \cite{cai2013graph}  (Lemma 4.3) also works in the planar setting.
So, we have the following

\begin{lemma}\label{lemma: full_connectedHard}
    For a $q \times q$ reducible symmetric matrix such that $M = A \oplus B$, $\PlEVAL(M)$ is $\#$P-hard iff at least one of  $\PlEVAL(A)$ or $\PlEVAL(B)$ is $\#$P-hard, and
    $\PlEVAL(M)$ is P-time tractable iff both  $\PlEVAL(A)$ or $\PlEVAL(B)$ are P-time tractable.
\end{lemma}

Specifically for the $q = 3$ case, we see that if $M$ is reducible,
$M$ is in form (\ref{eqn: full_disconnectedMatrix}), and $\PlEVAL(B)$
is trivially tractable, when $B$ is a $1 \times 1$ matrix.
Therefore, we have the following lemma.

\begin{lemma}\label{lemma: full_trivialCase}
    For $M$ in  form (\ref{eqn: full_disconnectedMatrix}),
    $\PlEVAL(M) \equiv \PlEVAL(M')$, where
    $M' = \left( \begin{smallmatrix}
        x & y\\
        y & z\\
    \end{smallmatrix} \right)$.
\end{lemma}

Thus,  \cref{theorem: full_domain2Hardness} already classifies $\PlEVAL(M)$
for reducible  matrices.

\subsection{Twinned matrices}\label{sec: full_twinned}
	
$\PlEVAL(M)$  is equivalent to a problem on $2 \times 2$ matrices
for another set of matrices $M$.
\begin{definition}\label{definition: full_twinnedMatrices}
	A $q \times q$ symmetric matrix $M$ is a twinned matrix if any of its rows (columns) is a multiple of another row (column). Specialized to $q = 3$, a symmetric matrix $M$ is a twinned matrix if it has the form
    \begin{equation}\label{eqn: full_twinnedMatrix}
        M^\pi = \begin{pmatrix}
		x & cx & y\\
		cx & c^{2}x & cy\\
		y & cy & z\\\end{pmatrix}
	\end{equation}
    after its rows and columns are permuted by a permutation $\pi$.
\end{definition}

Let $M' =  \begin{pmatrix}
	x & y\\
	y & z\\
	\end{pmatrix}$,
	and  $\mathcal{D} = \{D^{[r]}\}_{r \in \mathbb{N}}$, where $D^{[r]}$ denotes vertex weights
$D^{[r]} = \begin{pmatrix}
1 + c^r & 0\\
0 & 1\\
\end{pmatrix}$
for vertices of degree $r$.
Define $\PlEVAL(M', \mathcal{D})$ to
be the problem of evaluating
$$Z_{M', \mathcal{D}}(G) = \sum_{\sigma: V \rightarrow [2]} \prod_{\{u, v\} \in E}m'_{\sigma(u) \sigma(v)}\prod_{v \in V}D^{[\deg (v)]}_{\sigma(v) \sigma(v)},$$
for any planar graph $G = (V, E)$. 

Then for $M$ in  form (\ref {eqn: full_twinnedMatrix}),
\begin{align*}
     Z_{M}(G)
    &= \sum_{\sigma: V \rightarrow \{1, 2\}} \prod_{\{u, v\} \in E} m'_{\sigma(u) \sigma(v)} \left(\sum_{\tau: \sigma^{-1}(1) \rightarrow \{-, +\}}\left(\prod_{v \in \tau^{-1}(+)}(c^{\deg(v)}) \right)\right)\\
    &= \sum_{\sigma: V \rightarrow \{1, 2\}} \prod_{\{u, v\} \in E}m'_{\sigma(u) \sigma(v)}\prod_{v \in V}D^{[\deg (v)]}_{\sigma(v) \sigma(v)}\\
    &= Z_{M', \mathcal{D}}(G).
\end{align*}

Therefore,  $\PlEVAL(M) \equiv \PlEVAL(M', \mathcal{D})$. 
The following theorem is adapted from~\cite{govorov2020dichotomy}. 
(We give a proof for completeness in~\cref{sec: full_appendixA}.)

\begin{theorem}\label{theorem: full_degreeBoundedHardness}
For a full rank $M'$ and $\mathcal{D}$
given above, there exists $p_0 \ge 1$, such that  for all  $p\ge p_0$,
	 $\PlEVAL(N_{p}) \leq \PlEVAL(M', \mathcal{D})$
	 where
	$N_{p} = \begin{pmatrix}
	 c_p^{2}x &  c_py\\
	 c_p y & z
	\end{pmatrix}$ 
	and $c_p = \frac{1 + c^{2p + 1}}{1 + c^{2p}}$.
\end{theorem}

Before we prove the hardness for general twinned matrices, we do need to consider a special case where the problem is tractable.

\begin{definition}\label{definition: full_bipartiteMatrices}
	A $q \times q$ symmetric matrix $M$ is a bipartite matrix if it has the form $\begin{pmatrix}
	    \mathbf{0} & A\\
	    A^{\tt T} & \mathbf{0}\\
	\end{pmatrix}$ for some $r \times (q - r)$  matrix $A$, where $0<r<q$,
    after the rows and columns of $M$ are permuted by a permutation $\pi$.
\end{definition}

Specialized to $q = 3$, irreducible bipartite matrices happen to just be twinned matrices such that $x = z = 0$. It is known  that $\PlEVAL(M)$ is tractable in this case~\cite{bulatov2005complexity,goldberg2010complexity}.

\begin{lemma}\label{lemma: full_rankTwoMultipleNonOneC}
	Let $M$ be a twinned real-valued symmetric matrix
	in  form~(\ref{eqn: full_twinnedMatrix}). Assume $M$ is irreducible and non-bipartite.
	Then $\PlEVAL(M)$ is $\#$P-hard, unless $xz = y^{2}$, in which case, it is polynomial-time tractable.
\end{lemma}
\begin{proof}
	If $xz = y^{2}$ then $M$ is block rank $1$,
	and $\PlEVAL(M)$ is polynomial-time tractable~\cite{bulatov2005complexity,goldberg2010complexity}.
	Now we assume that this does not occur. 
	Note that $y \neq 0$ and $c \neq 0$, for otherwise $M$ would be reducible.
	Also, $(x, z) \neq \mathbf{0}$, for otherwise $M$ would be bipartite.

	Then for the $M'$ and $\mathcal{D}$
    given above, \cref{theorem: full_degreeBoundedHardness}
    applies, and we only need to prove that  $\PlEVAL(N_{p})$
    is \#P-hard for a  large positive integer $p$.
	
	We first assume  $c \ne \pm 1$.
    Then $c_p y \neq 0$ for any $p$, and also
    $\det(N_{p}) \neq 0$.
    It is easy to verify that 
    \[c_{p+1} - c_p = \frac{c^{2p}(c + 1)(c - 1)^{2}}{(1 + c^{2p})(1 + c^{2p + 2})}.
    \]
    Since $c \ne \pm 1$, 
    we have a strictly monotonic sequence $\{c_p \mid p \ge 1\}$.
    Hence for a large $p$, $c_{p}^2 x \ne \pm z$.
    Then $\PlEVAL(N_{p})$ is $\#$P-hard 
    by \cref{theorem: full_domain2Hardness}. It follows that $\PlEVAL(M)$ is also $\#$P-hard. 

    Now suppose $ c= \pm 1$.
    We have several cases.
     
     \begin{description}
         \item{\textbf{Case 1: $x^{2}z^{2} \neq y^{4}$ and $x^{2} \neq z^{2}$}}
         
        In this case we use $T_2M$.
        From \cref{theorem: full_degreeBoundedHardness}, we have $\PlEVAL(N) \leq \PlEVAL(T_{2}M) \leq \PlEVAL(M)$, where
    	$N = \left(\begin{smallmatrix}
    		x^{2} & y^{2}\\
    		y^{2} & z^{2}\\
    	\end{smallmatrix}\right)$.
        It follows from \cref{theorem: full_domain2Hardness} that $\PlEVAL(N)$ is $\#$P-hard, and so is $\PlEVAL(M)$.

        \item{\textbf{Case 2: $x^{2}z^{2} \neq y^{4}$ and $x^{2} = z^{2}$}}
        
        In this case we use $S_2T_2M$.
    	From \cref{theorem: full_degreeBoundedHardness}, we have $\PlEVAL(N') \leq \PlEVAL(S_{2}T_{2}M) \leq \PlEVAL(M)$, where
    	$N' = \left(\begin{smallmatrix}
    		2x^{4} + y^{4} & 3x^{2}y^{2}\\
    		3x^{2}y^{2} & x^{4} + 2y^{4}\\
    	\end{smallmatrix}\right)$.
        Since $x^{2}z^{2} \neq y^{4}$, the rank of $T_{2}M$ is $2$, and therefore, so is the rank of $S_{2}T_{2}M$. Therefore, it follows that $N'$ is also a rank $2$ matrix, and so, $\det(N') \neq 0$. Moreover, we note that since $x^{2} = z^{2}$, $x^{2}z^{2} = x^{4} \neq y^{4}$. Therefore, we see that $2x^{4} + y^{4} \neq x^{4} + 2y^{4}$. So $\PlEVAL(N')$ is $\#$P-hard
        from \cref{theorem: full_domain2Hardness} and so is $\PlEVAL(M)$.

        \item{\textbf{Case 3: $x^{2}z^{2} = y^{4}$ and $2x^{2} \neq y^{2}$}}
    
        In this case we use $S_2M$.
        As $xz \neq y^{2}$, $x^{2}z^{2} = y^{4}$ gives  $xz = -y^{2}$. Since $y \neq 0$, we have $x \neq 0$ and $z \neq 0$.
        Clearly $M = C^{\tt T} M' C$
        where $C = \left( \begin{smallmatrix}
    	1 & c & 0\\
    	0 & 0 & 1\\
    	\end{smallmatrix} \right)$ and
    	$M' = \left( \begin{smallmatrix}
    	x & y\\
    	y & z\\
    	\end{smallmatrix} \right)$. Then
        $S_{2}M$ has the matrix $M^2 = C^{\tt T} M'' C$,
    	where $M'' = M'  \left( \begin{smallmatrix}
    	2 & 0\\
    	0 & 1\\
    	\end{smallmatrix} \right) M'
    	= \left( \begin{smallmatrix}
    	x' & y'\\
    	y' & z'\\
    	\end{smallmatrix} \right)$,
    	with $x' = 2x^{2} + y^{2}, y' = 2xy + yz$ and $z' = 2y^{2} + z^2$.
    
        Clearly $M^2$ has a $2 \times 2$ submatrix 
        $M''$ with $\det(M'') \ne 0$.
        Also $x'z' + (y')^{2} >0$ since $y^2 >0$.
    	Therefore,  $(x')^{2}(z')^{2} - (y')^{4}  \neq 0$.  Also  $(x', z') \neq (0, 0)$. Finally, since $2x^{2} \neq y^{2}$, we have $y' \neq 0$, so $S_{2}M$ is irreducible. Therefore, it follows from cases 1 and 2 that $\PlEVAL(S_{2}M)$ is $\#$P-hard. Since $\PlEVAL(S_{2}M) \leq \PlEVAL(M)$, so is $\PlEVAL(M)$.
	
        \item{\textbf{Case 4: $x^{2}z^{2} = y^{4}$ and $2x^{2} = y^{2}$}}
    
        In this case we use $T_3M$.
        Again we have $xz = -y^{2}$. So, $z = - \frac{y^{2}}{x} = - \frac{2x^{2}}{x} = -2x$. 
        Then
    	$T_{3}M = \left(\begin{smallmatrix}
    		1 & c & 0\\
    		0 & 0 & 1\\
    	\end{smallmatrix}\right)^{\tt T}
    	\left(\begin{smallmatrix}
    		x' & y'\\
    		y' & z'\\
    	\end{smallmatrix}\right) \left(\begin{smallmatrix}
    		1 & c & 0\\
    		0 & 0 & 1\\
    	\end{smallmatrix}\right)$,
    	where
    	$x' = x^{3}, y' = y^{3}$ and $ z' = -8x^{3}$.
        Therefore, $T_{3}M$ is irreducible, $(x', z') \neq (0, 0)$, and $x'z' \neq (y')^{2}$. Also,
    	$(y')^{2} = y^{6} = 8x^{6} \neq 2x^{6} = 2(x')^{2}$.
    	Therefore, this case is reduced to case 3, and $\PlEVAL(T_{3}M)$ is therefore $\#$P-hard. Since $\PlEVAL(T_{3}M) \leq \PlEVAL(M)$, it implies that $\PlEVAL(M)$ is $\#$P-hard.
    \end{description}\vspace*{-\baselineskip}
\end{proof}

Combining the results from \cref{lemma: full_rankTwoMultipleNonOneC}
and \cref{theorem: full_domain2Hardness} we have

\begin{theorem}\label{theorem: full_rankTwoMultipleC}
If $M$ is a twinned real-valued symmetric matrix
	in  form~(\ref{eqn: full_twinnedMatrix}), then	
	$\PlEVAL(M)$ is $\#$P-hard, except in the following cases where $\PlEVAL(M)$ is polynomial-time tractable:
	\[\mbox{(1)~} y = 0,~~
	\mbox{(2)~} x = z = 0,~~
	\mbox{(3)~} xz = y^{2},~~
	\mbox{(4)~} c = 0 ~\&~ x = z ~~\mbox{or}~~
	\mbox{(5)~} c = 0 ~\&~ x = -z ~\&~ xz = -y^2. \]
\end{theorem}

\section{Interpolation by Thickening}\label{sec: full_interpolationThickening}

We have successfully classified as polynomially tractable or $\#$P-hard, all problems $\PlEVAL(M)$ for which $\PlEVAL(M)$ was equivalent to a domain two problem. We will now consider the problems for which such equivalences do not hold.
In this section we will furthermore not restrict ourselves to $q = 3$, but instead
consider the problem more generally.

Let us now consider the thickening operation more closely. We note that
\begin{equation}\label{eqn: full_thickening}
    Z_{M}(T_{k}G) = Z_{T_{k}M}(G) = \sum_{x \in X}x^{k} \cdot \text{\#}_{M}(G, x)
\end{equation}
where $X$ is as in \cref{full_definitionreuse}. Note that $\text{\#}_{M}(G, x)$ does not depend on $k$, but depends on the entries of the matrix $M$. We will deal with this dependence now.

\subsection{Generating sets}\label{sec: full_generatingSets}

\begin{definition}\label{definition: full_generatingSet}
	Let $\mathcal{A}$ be a set of nonzero
	real numbers. A finite set of 
	positive real numbers $\{g_{t}\}_{t \in [d]}$, for some integer $d \geq 0$, is a generating set of $\mathcal{A}$ if
for every $a \in \mathcal{A}$, there exists a \emph{unique} $(e_{1}, \dots, e_{d}) \in {\mathbb{Z}}^{d}$ such that
		$a = \pm {g_{1}^{e_{1}} \cdots g_{d}^{e_{d}}}$.
\end{definition}

\begin{lemma}\label{lemma: full_generatingSet}
Every finite set $\mathcal{A}$ of nonzero real numbers has
a generating set.
\end{lemma}
\begin{proof}
Consider the multiplicative group $\mathcal{G}$ 
generated by 
the positive real numbers 
$\{|a| : a \in \mathcal{A}\}$.
It is  a subgroup of the multiplicative group $(\mathbb{R}^+, \cdot)$.
Since  $\mathcal{A}$ is finite, and  $(\mathbb{R}^+, \cdot)$ is torsion-free,
the group $\mathcal{G}$ is  a finitely generated free Abelian group, 
and thus isomorphic to 
$\mathbb{Z}^d$ for some $d \ge 0$.
Let  $\{g_{t}\}_{t \in [d]}$ be a basis of this free Abelian group,
the lemma follows.
\end{proof}

We note that $\{g_{t}\}_{t \in [d]} \subset \mathcal{G}$ 
and $\{\log g_{t}\}_{t \in [d]}$ is linearly independent over $\mathbb{Q}$.



We now use \cref{lemma: full_generatingSet} to find a generating set 
for the entries $(m_{ij})_{i, j \in [q]}$ of any matrix $M \in \mathbb{R}^{q \times q}$ with no zero entries. 
Note that this generating set need not be unique. 
With respect to a generating set, for any $m_{ij}$, 
there are unique integers $e_{ij0} \in \{0, 1\}$, $e_{ij1}, \dots, e_{ijd}$, such that
\begin{equation}\label{eqn: full_generatingM}
    m_{ij} = (-1)^{e_{ij0}} \cdot g_{1}^{e_{ij1}} \cdots g_{d}^{e_{ijd}}.
\end{equation}

We also note that $\PlEVAL(M) \equiv \PlEVAL(cM)$ 
for any real $c \neq 0$, since $Z_{cM}(G) = c^{|E(G)|}Z_{M}(G)$. 
By choosing some $c = g_{1}^{e_{1}'} \cdots g_{d}^{e_{d}'}$
we may assume that 
$e_{ijt} \geq 0$ for all $i, j \in [q]$ and $t \in [d]$.

\begin{lemma}\label{lemma: full_thickeningInterpolation}
	Let $M \in \mathbb{R}^{q \times q}$ be  symmetric with no zero entries,
	with entries $(m_{ij})_{1 \leq i \leq j \leq q}$ 
	given  in \cref{eqn: full_generatingM}. 
	Define
	$\mathcal{M}: \mathbb{R}^{d} \rightarrow \mathbb{R}^{q \times q}$ where
	$\mathcal{M}(\mathbf{p})_{ij} = \mathcal{M}(p_{1}, \dots, p_{d})_{ij} = (-1)^{e_{ij0}} \cdot p_{1}^{e_{ij1}} \cdots p_{d}^{e_{ijd}}$ 
	for all $i, j \in [q]$. 
	Then, $\PlEVAL(\mathcal{M}(\mathbf{p})) \leq \PlEVAL(M)$ 
	for all $\mathbf{p} \in \mathbb{R}^{d}$.
\end{lemma}
\begin{proof}
    We already know that for any $k \geq 1$,
    $$Z_{M}(T_{k}G) = \sum_{x \in X}x^{k} \cdot \text{\#}_{M}(G, x),$$
    where
    $$X = \left\{\prod_{i,j\in [q]}m_{ij}^{k_{ij}}\hspace{0.08cm}\Big|\hspace{0.1cm}
    \text{integers $k_{ij}\ge 0$ and $\sum_{i,j\in [q]}k_{ij}=|E|$} \right\}.$$
    Since each $x \in X$ is distinct, and $|X| \leq |E|^{O(1)}$, 
    if we can compute $Z_{M}(T_{k}G)$ for $k \in [|X|]$, 
    we have a full rank Vandermonde system of linear equations, 
    which can be solved in polynomial time 
    to find $\text{\#}_{M}(G, x)$ for all $x \in X$.
    
    Now, let us consider the set $X$ more closely.
    Given any $x \in X$, we see that $x = \prod m_{ij}^{k_{ij}}$
    for some (not necessarily unique) integers $k_{ij} \geq 0$ such that $\sum k_{ij} = |E|$. 
    From \cref{eqn: full_generatingM}, we know that 
    each $m_{ij}$ is generated by the set $\{g_{t}\}_{t \in [d]}$.
    Therefore, any $x \in X$ can be represented as
    $$x = (-1)^{e^{x}_{0}}g_{1}^{e^{x}_{1}} \cdots g_{d}^{e^{x}_{d}}$$
    Moreover, the exponents $e^{x}_{0} \in \{0, 1\}$, 
    and $e^{x}_{1}, \dots, e^{x}_{d} \in \mathbb{Z}$ are unique, 
    since $\{g_{t}\}_{t \in [d]}$ is a generating set.
    
    Consider a fixed $\mathbf{p} = (p_{1}, \dots, p_{d}) \in \mathbb{R}^{d}$. We can now define the function $y: X \rightarrow \mathbb{R}$, such that $y(x) = (-1)^{e^{x}_{0}} \cdot p_{1}^{e^{x}_{1}} \cdots p_{d}^{e^{x}_{d}}$ for all $x \in X$.
    Now, let
    $$Y = \left\{\prod_{i,j\in [q]}(\mathcal{M}(\mathbf{p})_{ij}^{k_{ij}}\hspace{0.08cm}\Big|\hspace{0.1cm}
    \text{integers $k_{ij}\ge 0$ and $\sum_{i,j\in [q]}k_{ij}=|E|$} \right\}.$$
    
    We note that for any $y \in Y$, $$\text{\#}_{\mathcal{M}(\mathbf{p})}(G, y) = \sum_{x \in X: ~y(x) = y}\text{\#}_{M}(G, x).$$
    
    Since $\text{\#}_{M}(G, x)$ have already been computed for all $x \in X$, we can now compute
     in polynomial time,
    $$\sum_{x \in X}y(x) \cdot \text{\#}_{M}(G, x) = \sum_{y \in Y}y \cdot \text{\#}_{\mathcal{M}(\mathbf{p})}(G, y) = Z_{\mathcal{M}(\mathbf{p})}(G).$$
    Therefore, $\PlEVAL(\mathcal{M}(\mathbf{p})) \leq \PlEVAL(M)$.
\end{proof}


We need the following theorem from \cite{vertigan2005computational}:


\begin{theorem}\label{theorem: full_tutte}
	For $x, y \in \mathbb{C}$, evaluating the Tutte polynomial 
	at $(x, y)$ is $\#$P-hard over planar graphs unless 
	$(x - 1)(y - 1) \in \{1, 2\}$ or $(x, y) \in \{(1, 1), (-1,-1), (\omega,\omega^{2}), (\omega^{2}, \omega)\}$, 
	where $\omega = e^{\nicefrac{2\pi i}{3}}$. 
	In each exceptional case, the problem is 
	in polynomial time.
\end{theorem}
\begin{corollary}\label{corollary: hardness}
	$\PlEVAL(\textsc{VC}_{q})$ is $\#$P-hard for $q \geq 3$, 
	where $\textsc{VC}_{q}$ is the $q \times q$ matrix 
	with entries $(v_{ij})$ such that $v_{ij} = 1$ if $i \neq j$, and $v_{ij} = 0$ otherwise.
\end{corollary}

\cref{theorem: full_tutte} allows us to prove our first hardness result.

\begin{lemma}\label{lemma: full_thickeningBasic}
	Let $M$ be a $q \times q$ ($q \geq 3$) real-valued,
	symmetric matrix with no zero entries, as given in
	\cref{eqn: full_generatingM}. 
	Furthermore, assume for all $i \in [q]$ there exists some 
	(not necessarily distinct) $t(i) \in \{1, \dots, d\}$, 
	such that $e_{iit(i)} > 0, \text{ and } e_{jkt(i)} = 0$ for all $ j \neq k$.
	Then $\PlEVAL(M)$ is $\#$P-hard.
\end{lemma}
\begin{proof}
	We apply \cref{lemma: full_thickeningInterpolation}.
	Let $\mathbf{p} \in \mathbb{R}^{d}$,
	defined by  $p_{t(i)} = 0$ for $i \in [q]$,
	and $p_{i} = 1$ for all other $i \in [d]$.
	Then, it is easy to see that
	$T_{2}(\mathcal{M}(\mathbf{p})) = \textsc{VC}_{q}$.
	
	From \cref{lemma: full_thickeningInterpolation}, we get  $\PlEVAL(\mathcal{M}(\mathbf{p})) \leq \PlEVAL(M)$. Therefore,
	$$ \PlEVAL(\textsc{VC}_{q}) \leq \PlEVAL(\mathcal{M}(\mathbf{p})) \leq \PlEVAL(M).$$
	It follows from  \cref{corollary: hardness} that 
	$\PlEVAL(M)$ is  $\#$P-hard.
\end{proof}

\section{Interpolation by Stretching}\label{sec: full_interpolationStretching}

In this section we focus on full ranked matrices.
Using stretching,
we shall prove the hardness of a more
interesting class of matrices than we were able to
in \cref{lemma: full_thickeningBasic}.

Consider a $q \times q$ positive real valued, symmetric matrix $M$.
There exist real orthogonal matrix $H = (h_{ij})_{i, j \in [q]}$  and
 a diagonal matrix $D =  \text{diag}(\lambda_{1}, \dots, \lambda_{q}),$
such that
$$M = HDH^{\tt T},$$
where $\lambda_{1}, \dots, \lambda_{q}$ are the eigenvalues of $M$
and  the columns of $H$  are the corresponding eigenvectors.

\subsection{Lattice condition}\label{sec: full_latticeCondition}

From the decomposition  $M = HDH^{\texttt{T}}$, we have  $M^k = HD^kH^{\texttt{T}}$,
and
$$(M^{k})_{ij} = (h_{i1}h_{j1}) \lambda_{1}^k + \dots + (h_{iq}h_{jq}) \lambda_{q}^{k}.$$
It follows that
\begin{equation}\label{eqn: full_stretching}
    Z_{M^k}(G) = \sum_{\substack{x_{1}, \dots, x_{q} \ge 0\\ \sum_i x_{i} = |E|}}c_{(x_{i})_{i \leq q}} \cdot \left(\lambda_{1}^{x_{1}} \cdots \lambda_{q}^{x_{q}} \right)^{k},
\end{equation}
where
$$c_{(x_{i})_{i \leq q}} = \sum_{\sigma: V \rightarrow [q]} \left( \sum_{\substack{E_{1} \sqcup \dots \sqcup E_{q} = E\\|E_{i}| = x_{i}}} \left( \prod_{i \in [q]} \prod_{\{u, v\} \in E_{i}}h_{\sigma(u)i}h_{\sigma(v)i} \right) \right)$$
depends only on $G$ and the orthogonal matrix $H$, but  not  on $D$.

\begin{definition}\label{definition: full_lattice}
    A nonempty set of nonzero real numbers 
	$(r_{i})_{i \in [d]}$ satisfies the lattice condition, 
	if the only integer sequence $(n_{i})_{i \in [d]}$ 
	with the property $n_{1} + \dots + n_{d} = 0$ 
	and $r_{1}^{n_{1}} \cdots r_{d}^{n_{d}} = 1$ is 
	$(n_{i})_{i \in [d]} = \mathbf{0}$.
\end{definition}

\begin{lemma}\label{lemma: full_stretchingBasic}
	If $M$ is a $q \times q$ full rank, real valued, symmetric matrix, whose eigenvalues $(\lambda_{1}, \dots, \lambda_{q})$ satisfy the lattice condition 
	then $\PlEVAL(H\Delta H^{\tt{T}}) \leq \PlEVAL(M)$ 
	for any diagonal matrix $\Delta$.
\end{lemma}
\begin{proof}
	By the lattice condition, 
	for any integer sequences $(x_{i})_{i \leq q}$ and $(y_{i})_{i \leq q}$
	with 
	$\prod_{i \in [q]}\lambda_{i}^{x_{i}} = \prod_{i \in [q]}\lambda_{i}^{y_{i}}$ 
	and $\sum_{i \leq q}x_{i} =  \sum_{i \leq q}y_{i}$ we get  
	$x_{i} = y_{i}$ for all $i \in [q]$. 
	Therefore, from the values $Z_{M}(S_{k}G) = Z_{M^k}(G)$ 
	for $k \in \left[\binom{|E|+ q - 1}{q - 1}\right]$,
	we have a full-rank Vandermonde system of linear equations with
	unknowns $c_{(x_{i})_{i \leq q}}$. Solving this linear system
	in polynomial time,
	we can compute
	$$\sum_{\substack{x_{1}, \dots, x_{q}\\ \sum x_{i} = |E|}}c_{(x_{i})_{i \leq q}} \cdot \left(\alpha_{1}^{x_{1}} \cdots \alpha_{q}^{x_{q}} \right) = Z_{H\Delta H^{\texttt{T}}}(G)$$
	for any $\Delta = \text{diag}(\alpha_{1}, \dots, \alpha_{q})$.
	Therefore, $\PlEVAL(H\Delta H^{\texttt{T}}) \leq \PlEVAL(M)$.
\end{proof}

We now prove that there exists some $\Delta$ such that $\PlEVAL(H\Delta H^{\texttt{T}})$ is $\#$P-hard.

\begin{lemma}\label{lemma: full_latticeHardnessDelta}
	If $M = HDH^{\texttt{T}}$ is a $q \times q$ ($q \geq 3$) full rank, positive real valued, symmetric matrix, whose eigenvalues $(\lambda_{1}, \dots, \lambda_{q})$ satisfy the lattice condition, then there exists a diagonal matrix $\Delta$ such that $\PlEVAL(H\Delta H^{\texttt{T}})$ is $\#$P-hard.
\end{lemma}
\begin{proof}
    Let $(m_{ij})_{i, j \in [q]}$ be the entries of the matrix $M$. 
    By assumption, we know that these are positive reals.
	We pick some  
	$\kappa \in \mathbb{R}$ such that $\kappa + m_{ii} > 0$,
	and is	transcendental to the field 
	${\bf F} = \mathbb{Q}(\{m_{ij}\}_{i, j \in [q]})$.
	Such a 	transcendental real number exists because there are only a countable number of
	algebraic numbers over ${\bf F}$.
    Let $\{g_{i}\}_{i \in [d]}$ (respectively, $\{f_{i}\}_{i \in [d']}$) be a basis
    of the multiplicative free Abelian group generated by
    $\{m_{ij}\}_{i \neq j \in [q]}$ 
    (respectively, $\{\kappa + m_{ii}\}_{i \in [q]}$) as in the proof of 
    \cref{lemma: full_generatingSet}. 
    Finally, we  
    let $\Delta = D + \kappa I$.
	
	\begin{restatable}{claim}{generatingSetTrue}\label{claim: generatingSetTrue}
	    $\{g_{i}\}_{i \in [d]} \cup \{f_{i}\}_{i \in [d']}$ 
        is a generating set of the entries 
        of $H\Delta H^{\texttt{T}}$.
    \end{restatable}
    
    Clearly, every element of $H \Delta H^{\texttt{T}}$ 
    can be expressed as a product of
    integer powers of $\{g_{i}\}_{i \in [d]} \cup \{f_{i}\}_{i \in [d']}$, by construction.
    We now want to show uniqueness of such an expression.
 
    Being in the Abelian group generated by
    $\{\kappa + m_{jj}\}_{j \in [q]}$, there exist 
    integers $x_{i, j}$ 
    for $i \in [d']$ and $j \in [q]$, such that
    \begin{equation}\label{eqn: full_f_i-by-deltajj}
        f_{i} = (\kappa + m_{11})^{x_{i, 1}} \cdots ( \kappa + m_{qq})^{x_{i, q}}.
    \end{equation}
    Every element in 
    $\{g_{i}\}_{i \in [d]} \cup \{f_{i}\}_{i \in [d']}$ is
    positive. Suppose 
    for some 
    $(e_{1}, \dots, e_{d}, e'_{1}, \dots, e'_{d'}) \in \mathbb{Z}^{d + d'}$ 
    such that
    $$g_{1}^{e_{1}} \cdots g_{d}^{e_{d}} \cdot (f_{1})^{e'_{1}} \cdots (f_{d'})^{e'_{d'}} = 1.$$
    First, if $(e_{1}, \dots, e_{d}) = \mathbf{0}$ 
    then $\prod_{i \in [d]}g_{i}^{e_{i}}=1$, and since $\{f_{i}\}_{i \in [d']}$ is a generating set 
    we get $(e'_{1}, \dots, e'_{d'}) = \mathbf{0}$, 
    therefore $(e_{1}, \dots, e_{d}, e'_{1}, \dots, e'_{d'}) = \mathbf{0}$. 
    Now assume  
    $(e_{1}, \dots, e_{d}) \neq \mathbf{0}$.
    
    Substituting $f_i$ using \cref{eqn: full_f_i-by-deltajj}, we get
    $$\prod_{i \in [d]}g_{i}^{e_{i}} \cdot \prod_{j \in [q]} ( \kappa + m_{jj} )^{y_{j}} = 1,$$
    where $y_{j} = \sum_{i \in [d']}e'_{i}x_{i, j}$ 
    for $j \in [q]$.
    Since $(e_{1}, \dots, e_{d}) \neq \mathbf{0}$, 
    we see that $\prod_{i \in [d]}g_{i}^{e_{i}} \neq 1$. 
    Therefore, 
    $(y_{1}, \dots, y_{q}) \ne  \mathbf{0}$. 
    Separating out positive and negative $y_j$'s, we have
    \begin{equation}\label{eqn: full_generatingSetDecomp}
        \left(\prod_{i \in [d]}g_{i}^{e_{i}}\right) \cdot \prod_{j \in [q]: y_{j} > 0}(\kappa + m_{jj} )^{y_{j}} = \prod_{j \in [q]: y_{j} < 0}(\kappa + m_{jj} )^{-y_{j}}.
    \end{equation}
    Both sides of \cref{eqn: full_generatingSetDecomp} are polynomials in $\kappa$ over the field ${\bf F}$, with different leading coefficients. 
    This contradicts our assumption that $\kappa$ is 
    transcendental to ${\bf F}$.
    \cref{claim: generatingSetTrue} is thus proved.

    Now for any $i \in [q]$, 
    there exists some $t(i) \in [d']$, such that $e_{iit(i)} > 0$, 
    but $e_{jkt(i)} = 0$ for all $j \neq k$. 
    This is because 
    $\{f_{i}\}_{i \in [d']}$, without
    $\{g_{i}\}_{i \in [d]}$, is a generator set
    for $\{\kappa + m_{ii}\}_{i \in [q]}$,
    and 
    $\kappa + m_{ii} \ne 1$. Also,
    $\{g_{i}\}_{i \in [d]}$, without  $\{f_{i}\}_{i \in [d']}$, is a generator set
    for $\{m_{ij}\}_{i \neq j \in [q]}$.
    Therefore, from \cref{lemma: full_thickeningBasic}, 
    we conclude that $\PlEVAL(H\Delta H^{\texttt{T}})$ is $\#$P-hard.
\end{proof}

We have prove the following theorem:
\begin{theorem}\label{theorem: full_latticeHardness}
	If $M$ is a $q \times q$ ($q \geq 3$) full rank, positive real valued, symmetric matrix, whose eigenvalues $(\lambda_{1}, \dots, \lambda_{q})$ satisfy the lattice condition, then $\PlEVAL(M)$ is $\#$P-hard.
\end{theorem}
\subsection{Extensions of the hardness criterion}\label{sec: full_hardnessExtension}
	
The requirement in \cref{theorem: full_latticeHardness}
that the eigenvalues satisfy the lattice condition
is not entirely necessary. 
The following is a simple adaptation of \cref{lemma: full_stretchingBasic}.

\begin{lemma}\label{lemma: full_stretchingExtension}
    If $M$ is a $q \times q$ full rank, real valued, symmetric matrix, such that its set of eigenvalues $\{\lambda_{i}: i \in [q]\}$ (without duplicates as a set) satisfies the lattice condition,
	then for any function $f: \{\lambda_{i}: i \in [q]\} \rightarrow \mathbb{R}$, we have $\PlEVAL(H\Delta H^{\tt{T}}) \leq \PlEVAL(M)$ 
	where
	$\Delta = {\rm diag}(f(\lambda_{1}), \dots, f(\lambda_{q}))$.
\end{lemma}
\begin{proof}
Note that as a function, if $\lambda_{i} = \lambda_{j}$ then $f$ must map
$f(\lambda_{i}) = f(\lambda_{j})$.  Accordingly, 
 we can define a partition $\mathcal{P} = (P_{1}, \dots, P_{\ell})$ of $[q]$ collecting
 equal values of $\lambda_{i}$ together.
We rename their distinct values as $\{\mu_{1}, \dots, \mu_{\ell}\} $
    such that the $\mu_{i}$ are all distinct, and
    $\lambda_{j} = \mu_{i}$ for all $i \in [\ell]$ and $j \in P_{i}$.
By hypothesis, $(\mu_{1}, \dots, \mu_{\ell})$ satisfies the lattice condition, and $f$ is defined on the set $\{\mu_{i}: i \in [\ell]\}$.
    
	Now note that
	$$Z_{M}(G) = \sum_{\substack{x_{1}, \dots, x_{\ell}\\ \sum x_{i} = |E|}}c_{(x_{i})_{i \leq \ell}} \cdot \mu_{1}^{x_{1}} \cdots \mu_{\ell}^{x_{\ell}}$$
	where
	$$c_{(x_{i})_{i \leq \ell}} = \sum_{\sigma: V \rightarrow [q]} \left( \sum_{\substack{E_{1} \sqcup \dots \sqcup E_{q} = E,\\ \sum_{t \in P_{i}}|E_{t}| = x_{i}}} \left( \prod_{i \in [q]} \prod_{\{u, v\} \in E_{i}} h_{\sigma(u)i}h_{\sigma(v)i} \right) \right).$$
	
	Since $(\mu_{1}, \dots, \mu_{\ell})$ satisfies the lattice condition, 
	for any $(x_{i})_{i \leq \ell}$ and $(y_{i})_{i \leq \ell}$ with $\sum_i x_i = \sum_i y_i = |E|$, and
	$$\mu_{1}^{x_{1}} \cdots \mu_{\ell}^{x_{\ell}} = \mu_{1}^{y_{1}} \cdots \mu_{\ell}^{y_{\ell}},$$
	we have $x_{i} = y_{i}$ for $i \in [\ell]$. 
	Therefore, from the values $Z_{M}(S_{k}G) = Z_{M^k}(G)$ 
	for $k \in \left[\binom{|E| + \ell - 1}{\ell - 1}\right]$, 
	we can form a full rank Vandermonde system of linear equations
	with unknowns $c_{(x_{i})_{i \leq \ell}}$.
	Solving this in polynomial time, 
	we can compute
	$$\sum_{\substack{x_{1}, \dots, x_{\ell}\\ \sum x_{i} = |E|}}c_{(x_{i})_{i \leq \ell}} \cdot f(\mu_{1})^{x_{1}} \cdots f(\mu_{\ell})^{x_{\ell}} = Z_{H \Delta H^{\tt T}}(G)$$
	for any $\Delta = \text{diag}(f(\lambda_{1}), \dots, f(\lambda_{q}))$. 
	Therefore, $\PlEVAL(H\Delta H^{\tt T}) \leq \PlEVAL(M)$.
\end{proof}

We now have the following extension of \cref{theorem: full_latticeHardness}.

\begin{theorem}\label{theorem: full_latticeHardnessExtension}
	If $M$ is a $q \times q$  ($q \ge 3$) full rank, positive real valued, symmetric matrix, such that its set of eigenvalues $\{\lambda_{i}: i \in [q]\}$ (without duplicates as a set) satisfies the lattice condition, then $\PlEVAL(M)$ is $\#$P-hard.
\end{theorem}
\begin{proof}
    Let $(m_{ij})_{i, j \in [q]}$ be the entries of the matrix $M$. 
    By assumption, we know that these are positive reals.
	We pick some  
	$\kappa \in \mathbb{R}$ such that $\kappa + m_{ii} >0$,
	and is	transcendental to the field 
	${\bf F} = \mathbb{Q}(\{m_{ij}\}_{i, j \in [q]})$.
    Let $\{g_{i}\}_{i \in [d]}$ be a basis of the
    multiplicative free Abelian group generated by
    $\{m_{ij}\}_{i \neq j \in [q]}$ as in the
    proof of \cref{lemma: full_generatingSet}.
    Let $\{f_{i}\}_{i \in [d']}$ be a basis for the
    multiplicative free Abelian group
    generated by $\{\kappa + m_{ii}\}_{i \in [q]}$.
    Finally, we let $\Delta = D + \kappa I$.
    We know from \cref{claim: generatingSetTrue} 
    that the set $\{g_{i}\}_{i \in [d]} \cup \{f_{i}\}_{i \in [d']}$ 
    is a generating set of the entries of $H \Delta H^{\tt T} = M + \kappa I$.
    
    Now for any $i \in [q]$, 
    there exists some $t(i) \in [d']$, 
    such that $e_{iit(i)} > 0$, but $e_{jkt(i)} = 0$ 
    for all $i \neq k$. 
    Therefore, from \cref{lemma: full_thickeningBasic}, 
    we conclude that $\PlEVAL(H\Delta H^{\tt T})$ is $\#$P-hard.
    Moreover, $\Delta$ is of the form
    $\text{diag}(f(\lambda_{1}), \dots, f(\lambda_{q}))$ with 
    $f(\lambda_{i}) = \lambda_{i} + \kappa$.
    So, from \cref{lemma: full_stretchingExtension}, 
    we see that $\PlEVAL(M)$ is $\#$P-hard.
\end{proof}

Clearly, \cref{theorem: full_latticeHardness} is a special case of \cref{theorem: full_latticeHardnessExtension}.
We also have the following theorem.

\begin{theorem}\label{theorem: full_qTimesqNonBipLatticeNonNegative}
Let $M$ be a $q \times q$ ($q \geq 3$)  \emph{non-bipartite}, irreducible, full rank, non-negative real valued symmetric matrix. If its set of absolute values of eigenvalues 
$\{|\lambda_{i}|: i \in [q]\}$ (without duplicates as a set) 
satisfies the lattice condition, then  $\PlEVAL(M)$ is $\#$P-hard. 
\end{theorem}
\begin{proof}
    The matrix $M$ with non-negative values represents a weighted graph $H$. Since $M$ is irreducible, $H$ is a connected graph. Since $M$ is non-bipartite, $H$  contains a cycle of odd length $t$. Moreover, $H$ trivially has cycles of length $2$. Since $H$ is connected, and $\gcd(t, 2) = 1$, for any large enough integer $n$, there is a path of length $n$ between any $i, j \in [q]$. In other words, for some large enough integer $n$, $M^{n}$ is a full rank, positive valued matrix.  Since this is true for all sufficiently large $n$, we may assume $n$ is even.
    The eigenvalues of $M^n$ are $\lambda_{1}^{n}, \dots, \lambda_{q}^{n}$,
    with possible repetition. Suppose
    $\{\mu_1, \ldots, \mu_\ell\}$ is the set
    $\{|\lambda_{i}|: i \in [q]\}$ after  removing duplicates, 
    then by hypothesis it
    satisfies the lattice condition.
    Then $\{\mu_1^n, \ldots, \mu_\ell^n\}$ is the set
    $\{\lambda_{i}^n: i \in [q]\}$ without duplicates as a set. Indeed, if $\lambda_{i}^n = \lambda_{j}^n$
    then as real numbers $\lambda_{i} = \pm \lambda_{j}$,
    and so $|\lambda_{i}| = |\lambda_{j}|$. Thus only one of 
    $\lambda_{i}^n$ and $\lambda_{j}^n$ appears in $\{\mu_1^n, \ldots, \mu_\ell^n\}$.
   It follows that
    $\{\mu_1^n, \ldots, \mu_\ell^n\}$ also 
     satisfies the lattice condition. Therefore, 
     $\PlEVAL(S_{n}M)$ is $\#$P-hard by \cref{theorem: full_latticeHardnessExtension}. It follows that $\PlEVAL(M)$ is also $\#$P-hard.
\end{proof}

Next we prove the same theorem for the bipartite case. For the bipartite case,
we note that if a $q \times q$ matrix $M$ has full rank, then  $q$ is even and
after a permutation, $M$ has the form    $\left(\begin{smallmatrix} \mathbf{0} & A\\ A^{\tt T} & \mathbf{0} \end{smallmatrix}\right)$, for some matrix $A$ of order $q/2$.
That the  lattice condition implies \#P-hardness,
 as in \cref{theorem: full_qTimesqNonBipLatticeNonNegative}, really only works 
when $q/2 \ge 3$. 
So the following theorem is stated for $q \ge 6$. After the theorem, we give a
complete classification for such bipartite matrices with $q=4$.

\begin{theorem}\label{theorem: full_qTimesqBipartitieLatticeNonNegative}
Let $M$ be a $q \times q$ ($q \geq 6$)  \emph{bipartite}, irreducible, full rank, non-negative real valued symmetric matrix.  If the absolute values of the eigenvalues 
$\{|\lambda_{i}|: i \in [q]\}$,  as a set without duplicates, 
satisfies the lattice condition, then  $\PlEVAL(M)$ is $\#$P-hard. 
\end{theorem}
\begin{proof}
    For a $q \times q$ bipartite matrix $M$, its square
    $M^{2}$ is a reducible matrix of the form $M^2 = A \oplus B$. 
    Since $M$ has full rank $q$, $q$ must be even, and $A$ and $B$ are both $q/2 \times q/2$.
    Both 
  $A$ and $B$ are    non-negative real valued  symmetric matrices. As $M$ is irreducible
  and bipartite,
  the underlying graph is connected, and  every pair of vertices in each part is
  connected by path of an even length. Thus $A$ and $B$ are irreducible.
 Both $A$ and $B$ contain self loops and so they are non-bipartite. 
 They have full rank since $\det(M^2) = \det(A) \det(B)$.
 $\{\lambda_{i}^2: i \in [q]\}$ is the union of  eigenvalues of $A$ and $B$.
 Since $\{|\lambda_{i}|: i \in [q]\}$ (after removal of duplicates as a set) 
satisfies the lattice condition, so does the subset of $\{\lambda_{i}^2: i \in [q]\}$
that   corresponds to $A$ (and to $B$), both after removal of duplicates as a set.
As $q \ge 6$, we have $q/2 \ge 3$ and we can conclude that $\PlEVAL(A)$ is $\#$P-hard
by \cref{theorem: full_qTimesqNonBipLatticeNonNegative}. Then $\PlEVAL(M)$ is $\#$P-hard
by \cref{lemma: full_connectedHard}.
\end{proof}

The corresponding case  $q = 4$ in \cref{theorem: full_qTimesqBipartitieLatticeNonNegative}
can be completely classified. 
    In this case,  $M = \left(\begin{smallmatrix} \mathbf{0} & N\\ N^{\tt T} & \mathbf{0} \end{smallmatrix}\right)$,
    where $N = \left(\begin{smallmatrix} a & b \\ c & d \end{smallmatrix}\right)$ 
    is a full rank, non-negative real valued matrix.
    Now, we consider $M^{2} = A \oplus B$, where
    $$A = NN^{\tt T} = \begin{pmatrix}
        a^{2} + b^{2} & ac + bd\\
        ac + bd & c^{2} + d^{2}\\
    \end{pmatrix}, \qquad B = N^{\tt T} N = \begin{pmatrix}
        a^{2} + c^{2} & ab + cd\\
        ab + cd & b^{2} + d^{2}\\
    \end{pmatrix}.$$
    Since $M$ is irreducible, all entries of $A$ and $B$  are positive (as each side of
    the bipartite graph has a vertex connected to both vertices of the other side).
    Therefore, we see that  $\PlEVAL(M^2)$ is $\#$P-hard
    unless both
    $\PlEVAL(A)$ and $\PlEVAL(B)$ are tractable, and by \cref{theorem: full_domain2Hardness}
    this is so iff
    $a^{2} + b^{2} = c^{2} + d^{2}$ 
    and $a^{2} + c^{2} = b^{2} + d^{2}$.
    Since $N$ is a non-negative matrix, this implies that $a = d$ and $b = c$.
    Therefore, as long as $(a, b) \neq (d, c)$, at least one of $\PlEVAL(A)$ or $\PlEVAL(B)$ is $\#$P-hard, and this would imply the \#P-hardness of $\PlEVAL(M)$.
    
    If $(a, b) = (d, c)$, then  $M = X \otimes Y$, where $X = \left(\begin{smallmatrix} 0 & 1\\ 1 & 0 \end{smallmatrix}\right)$, and $Y = \left(\begin{smallmatrix} a & b\\ b & a \end{smallmatrix}\right)$.
    $\PlEVAL(X)$ and $\PlEVAL(Y)$ are tractable from \cref{theorem: full_domain2Hardness}.
    We note that for any planar graph $G = (V, E)$,
    \begin{align*}
        Z_{M}(G)
        &= \sum_{\sigma: V \rightarrow [4]}\prod_{\{u, v\} \in E}m_{\sigma(u)\sigma(v)}\\
        &= \sum_{(\sigma_{1}, \sigma_{2}): V \rightarrow [2] \times [2]}\prod_{\{u, v\} \in E}X_{\sigma_{1}(u)\sigma_{1}(v)}Y_{\sigma_{2}(u)\sigma_{2}(v)}\\
        &= \left(\sum_{\sigma_{1}: V \rightarrow [2]} \prod_{\{u, v\} \in E}X_{\sigma_{1}(u)\sigma_{1}(v)}\right) \left(\sum_{\sigma_{2}: V \rightarrow [2]} \prod_{\{u, v\} \in E}Y_{\sigma_{2}(u)\sigma_{2}(v)}\right)\\
        &= Z_{X}(G) \cdot Z_{Y}(G)
    \end{align*}
    Therefore, $\PlEVAL(M)$ is also polynomial time tractable. 
    
    Note that in this case, 
    the eigenvalues of $M$ are $\{\pm(a+b), \pm(a-b)\}$ 
    for real $a \ne \pm b$,
    with absolute values  $\{|a+b|, |a-b|\}$ after removal of duplicates.  
    One is greater than the other, 
    and so they \emph{do}  satisfy the lattice condition. Thus,
    the formal statement of \cref{theorem: full_qTimesqBipartitieLatticeNonNegative}
    for $q=4$ is false (assuming \#P does not collapse to P.)

\begin{theorem}\label{theorem: full_mesaure0}
    The set of  $q \times q$ ($q \geq 3$) real  symmetric matrices $M$ such that
    $\PlEVAL(M)$ is not $\#$P-hard has measure 0.
\end{theorem}
\begin{proof}
    Using the same proof idea of \cref{lemma: full_count} we can show that
    $\PlEVAL(|M|)  \leq \PlEVAL(M)$ where $|M|$ denotes the matrix 
    $(|m_{i,j}|)$ obtained from $M$ by taking entry-wise absolute values.
    Consider the set of  $q \times q$ non-negative real symmetric matrices.
    The subset that has rank $<q$ has measure 0. This is also the case for bipartite matrices and reducible
    matrices.
    The set of  lattice conditions  specified by an integer sequence  $(n_{i})_{i \in [d]} \ne {\bf 0}$ 
    in  \cref{definition: full_lattice} is 
    a countable set. Each such  condition defines a 
    hypersurface $n_1 \log(\lambda_{1}^2 ) + \cdots + n_d \log (\lambda_{d}^{2}) = 1$,
    where the eigenvalues $\lambda_{i}^2$ are continuous and  piece-wise differentiable
    functions of the entries of $M$. (If we order $\lambda_{1} \le  \ldots \le \lambda_{n}$,
    we can avoid a measure 0 subset where two eigenvalues are equal, which is specified by
    the vanishing of the discriminant, a polynomial in the entries of $M$.)
    Thus, the subset where the lattice condition fails is also of measure 0.
    It follows that $\PlEVAL(M)$ is  $\#$P-hard for almost all $M$ in the sense of
    Lebesgue measure.\footnote{The number of Turing machines is countable,
    and so there are only a countable number of algorithms. But this observation does not trivialize
    \cref{theorem: full_mesaure0}, since it is possible (and indeed true) that a single TM can solve
    uncountably many problems $\PlEVAL(M)$, by the strict definition of  $\PlEVAL(M)$ for all real $M$.}
\end{proof}

\section{Hardness of \texorpdfstring{$3 \times 3$}{3 x 3} matrices}\label{sec: full_positiveHardness}


Consider a $q \times q$ full rank, positive real valued, symmetric matrix $M$, with a generating set 
$\{g_{t}\}_{t \in [d]}$ 
obtained as in \cref{lemma: full_generatingSet}, and let 
$\mathcal{M}: \mathbb{R}^{d} \rightarrow \mathbb{R}^{q \times q}$ 
be defined as in \cref{lemma: full_thickeningInterpolation}.

\begin{lemma}\label{lemma: full_reductiontoMp}
	For any  nonzero 
	polynomial $f(x_1, \ldots, x_d) \in \mathbb{Z}[x_1,\ldots, x_d]$,
	there exist  nonnegative integers $e_{1}, \ldots, e_{d}$, such that
	$f(x^{e_1}, \ldots, x^{e_d}) \in \mathbb{Z}[x]$ is a nonzero 
	polynomial.
\end{lemma}
\begin{proof}
    If $d=0$, then $f$ is a nonzero integer, and the lemma is trivial.
    If $d=1$, the lemma is proved by taking $e_1=1$.
    Assume $d > 1$. 
    There exist $p, p_2, \ldots, p_d \in \mathbb{Z}$
    such that $f(p, p_2, \ldots, p_d) \not =0$. We may assume $p \ne -1, 0, 1$,
    since $f(x, p_2, \ldots, p_d) \in \mathbb{Z}[x]$ is a nonzero 
    polynomial, and  has only finitely many zeros. Then
    $f(p, x_2, p_3, \ldots, p_d) \in \mathbb{Z}[x_2]$ is a  nonzero univariate polynomial, which has finitely many zeros. Thus,
    for some integer $e_2 \ge 0$, 
    $f(p, p^{e_2}, p_3, \ldots, p_d) \ne 0$.
    Inductively, assume $f(p^1, p^{e_2},  \ldots, p^{e_{t-1}}, p_t, \ldots, p_d) \ne 0$, for some $t$,
    then $f(p^1, p^{e_2},  \ldots, p^{e_{t-1}}, x_t, p_{t+1}, \ldots, p_d) \in \mathbb{Z}[x_t]$ is a  nonzero univariate polynomial,  and thus for some
    integer $e_t \ge 0$, 
    $f(p^1, p^{e_2}, \ldots, p^{e_{t}},  p_{t+1}, \ldots, p_d) \ne 0$.
    Finally,  $f(p^1, p^{e_2},  \ldots, p^{e_d}) \ne 0$, and
    so the  univariate polynomial 
    $f(x^1, x^{e_2},  \ldots, x^{e_d}) \in \mathbb{Z}[x]$
    is nonzero.
\end{proof}
\begin{corollary}\label{cor: reductiontoMp-p-range}
	If $M$ is a $q \times q$ full rank, positive real valued, symmetric matrix, then there exist non-negative integers $e_{1}, \dots, e_{d}$, and  real $\epsilon > 0$, such that $\det\left(\mathcal{M}(p^{e_{1}}, \dots, p^{e_{d}})\right) \neq 0$, for all $1 < p < e^{\epsilon}$.
\end{corollary}
\begin{proof}
    $\det(\mathcal{M}(x_{1}, \dots, x_{d})) \in  \mathbb{Z}[x_1, \dots, x_d]$ is a nonzero polynomial since $\det(\mathcal{M}(g_{1}, \dots, g_{d})) = \det(M) \neq 0$. By \cref{lemma: full_reductiontoMp} we have a nonzero univariate polynomial $\det\left(\mathcal{M}(p^{e_{1}}, \dots, p^{e_{d}})\right)$.
    It has at most finitely many zeros, and so for some $\epsilon > 0$, 
    the value is nonzero for all $1 < p < e^{\epsilon}$.
\end{proof}

We will now focus our attention back on 
$3 \times 3$ matrices specifically, and prove the hardness 
of all full rank, positive real valued matrices.
We define the function $M: \mathbb{R} \rightarrow \mathbb{R}^{3 \times 3}$ as $M(p) := \mathcal{M}(p^{e_{1}}, \dots, p^{e_{d}})$, where $e_{1}, \dots, e_{d}$ are as in \cref{lemma: full_reductiontoMp}
and~\cref{cor: reductiontoMp-p-range}. Each entry of $M(p)$
has the form $M(p)_{ij} = p^{x_{ij}}$ for some
non-negative integer $x_{ij}$, and $x_{ij} = x_{ji}$.
We know from \cref{cor: reductiontoMp-p-range} that there exists some $\epsilon > 0$, such that $\det(M(p)) \neq 0$ for all $1 < p < e^{\epsilon}$.

\begin{lemma}\label{lemma: full_determinantNonZero}
    If $M$ is a $3 \times 3$ full rank, positive real valued symmetric matrix, with $M(p)_{ij} = p^{x_{ij}}$ for all $i, j \in [3]$
    and $x_{ij} = x_{ji}$, then $$\lim\limits_{\delta \rightarrow 0} \frac{\det M(e^{\delta}) }{\delta^{2}} = 0 ~~~\implies~~~ \lim\limits_{\delta \rightarrow 0} \frac{\det M(e^{\delta}) }{\delta^{3}} \neq 0.$$
\end{lemma}
\begin{proof}
   Consider the matrix
    $M(e^{\delta}) = (e^{x_{ij}\delta})_{i, j \in [3]}$.
    Define $X$ to be the $3 \times 3$ matrix with the entries $(x_{ij})_{i, j \in [3]}$, and consider the Taylor series expansion of $f(\delta) = \det M(e^{\delta})$,
    \begin{equation}\label{eqn: full_Taylor}
   f(\delta)= 
   f(0)+ f'(0) \delta 
   + f''(0) \frac{\delta^{2}}{2!} 
   + \big(3g(X) + 6 \det X \big)\frac{\delta^{3}}{3!} + O(\delta^{4}),
    \end{equation}
    where $f(0) = \det M(e^{0}) =0$ as $M(e^{0}) = J$ is the all-1 matrix, and
    $$ f'(0) = \begin{vmatrix}
        x_{11} & x_{12} & x_{13}\\
        1 & 1 & 1\\
        1 & 1 & 1\\
    \end{vmatrix} + \begin{vmatrix}
        1 & 1 & 1\\
        x_{12} & x_{22} & x_{23}\\
        1 & 1 & 1\\
    \end{vmatrix} + \begin{vmatrix}
        1 & 1 & 1\\
        1 & 1 & 1\\
        x_{13} & x_{23} & x_{33}\\
    \end{vmatrix} = 0,$$
    $$ \frac{1}{2}f''(0)  = \begin{vmatrix}
        x_{11} & x_{12} & x_{13}\\
        x_{12} & x_{22} & x_{23}\\
        1 & 1 & 1\\
    \end{vmatrix} + \begin{vmatrix}
        x_{11} & x_{12} & x_{13}\\
        1 & 1 & 1\\
        x_{13} & x_{23} & x_{33}
    \end{vmatrix} + \begin{vmatrix}
        1 & 1 & 1\\
        x_{12} & x_{22} & x_{23}\\
        x_{13} & x_{23} & x_{33}\\
    \end{vmatrix}, \text{ and }$$
    \begin{align*}
        g(X) &= \begin{vmatrix}
            (x_{11})^{2} & (x_{12})^{2} & (x_{13})^{2}\\
            x_{12} & x_{22} & x_{23}\\
            1 & 1 & 1\\
        \end{vmatrix} + \begin{vmatrix}
            1 & 1 & 1\\
            (x_{12})^{2} & (x_{22})^{2} & (x_{23})^{2}\\
            x_{13} & x_{23} & x_{33}\\
        \end{vmatrix} + \begin{vmatrix}
            x_{11} & x_{12} & x_{13}\\
            1 & 1 & 1\\
            (x_{13})^{2} & (x_{23})^{2} & (x_{33})^{2}\\
        \end{vmatrix}\\
        &\qquad + \begin{vmatrix}
            x_{11} & x_{12} & x_{13}\\
            (x_{12})^{2} & (x_{22})^{2} & (x_{23})^{2}\\
            1 & 1 & 1\\
        \end{vmatrix} + \begin{vmatrix}
            1 & 1 & 1\\
            x_{12} & x_{22} & x_{23}\\
            (x_{13})^{2} & (x_{23})^{2} & (x_{33})^{2}\\
        \end{vmatrix} + \begin{vmatrix}
            (x_{11})^{2} & (x_{12})^{2} & (x_{13})^{2}\\
            1 & 1 & 1\\
            x_{13} & x_{23} & x_{33}\\
        \end{vmatrix}.
    \end{align*}
    
    We remark that $f''(0) =0$ if ${\rm rank}~ X \le 1$.
    
    By the Taylor expansion, 
    $$\lim\limits_{\delta \rightarrow 0} \frac{\det M(e^{\delta}) }{\delta^{2}} = \frac{1}{2} f''(0).$$
After some row operations, we have
$$\frac{1}{2} f''(0) = 
        \begin{vmatrix}
            x_{11} - x_{13} & x_{12} - x_{23} & x_{13} - x_{33}\\
            x_{12} - x_{13} & x_{22} - x_{23} & x_{23} - x_{33}\\
            1 & 1 & 1\\
        \end{vmatrix}.$$
        Now we assume  $ f''(0)= 0$.
    So there exist real numbers $(a, b, c, d) \neq \mathbf{0}$, such that $a + b + c = 0$, and
    $$a \begin{pmatrix}
        x_{11} & x_{12} & x_{13}\\
    \end{pmatrix} + b \begin{pmatrix}
        x_{12} & x_{22} & x_{23}\\
    \end{pmatrix} + c \begin{pmatrix}
        x_{13} & x_{23} & x_{33}\\
    \end{pmatrix} + d \begin{pmatrix}
        1 & 1 & 1\\
    \end{pmatrix} = \mathbf{0}.$$
    Since $(a, b, c, d) \neq \mathbf{0}$, this equation
    also gives  $(a, b, c) \neq \mathbf{0}$. Therefore, we may assume without loss of generality that $c = -(a + b) \neq 0$. Let
    $\alpha = \frac{a}{a + b}$ and $\beta = \frac{d}{a + b}$, then 
    $$\begin{pmatrix}
        x_{13} & x_{23} & x_{33}\\
    \end{pmatrix} = \alpha \begin{pmatrix}
        x_{11} & x_{12} & x_{13}\\
    \end{pmatrix} + (1 -\alpha) \begin{pmatrix}
        x_{12} & x_{22} & x_{23}\\
    \end{pmatrix} + \beta \begin{pmatrix}
        1 & 1 & 1
    \end{pmatrix}.$$
This gives the expression $X =  A + \beta B$, 
where 
$$A = N^{\tt T} \begin{pmatrix}
        x_{11} & x_{12}\\
        x_{12} & x_{22}
    \end{pmatrix} N, ~~~ N = \begin{pmatrix} 
        1 & 0 & \alpha \\
        0 & 1 & 1- \alpha
    \end{pmatrix} , ~~~ B = \begin{pmatrix}
        0 & 0 & 1\\
        0 & 0 & 1\\
        1 & 1 & 2\\
    \end{pmatrix}.$$

    We will now make use of the following claim, which we shall prove later:
    \begin{restatable}{claim}{kNecessary}\label{claim: kNecessary}
        If $x_{11} + x_{22} = 2x_{12}$, then $\det M(e^{\delta}) = 0$ for all $\delta > 0$.
    \end{restatable}
    
    Since we have assumed that $M$ is a full rank matrix, we know from \cref{lemma: full_reductiontoMp} that for small enough values of $\delta$, $\det M(e^{\delta}) \neq 0$. Therefore, we have  $x_{11} + x_{22} \neq 2x_{12}$. Next we consider the matrix $A - kJ$, where
    $$k = \frac{x_{11}x_{22} - (x_{12})^{2}}{x_{11} + x_{22} - 2x_{12}}, 
    ~~~~J = \begin{pmatrix}
        1 & 1 & 1\\
        1 & 1 & 1\\
        1 & 1 & 1\\
    \end{pmatrix}.$$
    As $x_{11} + x_{22} \neq 2x_{12}$, the value $k$ is well-defined.
    The following claim will also be proved later:
    \begin{restatable}{claim}{kWorks}\label{claim: kWorks}
        The matrix $A - kJ$ has rank at most one.
    \end{restatable}

    This implies that there exists a vector
    $$\mathbf{u} = \begin{pmatrix}
        u_{1} & u_{2} & u_{3}\\
    \end{pmatrix}^{\tt T},$$
    such that $A = \mathbf{u} \mathbf{u}^{\tt T} + kJ$. Therefore,  $X = \mathbf{u} \mathbf{u}^{\tt T} + kJ + \beta B$. Next, we note that
    \begin{eqnarray}
        \addtocounter{equation}{-1}\refstepcounter{equation}\label{eqn:detM-delta-middle}
       \det M(e^{\delta}) 
        &= & \begin{vmatrix}
            e^{\delta((u_{1})^{2} + k)} & e^{\delta(u_{1}u_{2} + k)} & e^{\delta(u_{1}u_{3} + k + \beta)} \\ \nonumber
            e^{\delta(u_{1}u_{2} + k)} & e^{\delta((u_{2})^{2} + k)} & e^{\delta(u_{2}u_{3} + k + \beta)}\\
            e^{\delta(u_{1}u_{3} + k + \beta)} & e^{\delta(u_{2}u_{3} + k + \beta)} & e^{\delta((u_{3})^{2} + k + 2\beta)}
        \end{vmatrix}\\
        &= & e^{\delta(3k + 2\beta)}\begin{vmatrix}
            e^{\delta(u_{1})^{2}} & e^{\delta u_{1}u_{2}} & e^{\delta u_{1}u_{3}}\\
            e^{\delta u_{1}u_{2} } & e^{\delta(u_{2})^{2}} & e^{\delta u_{2}u_{3} }\\
            e^{\delta u_{1}u_{3} } & e^{\delta u_{2}u_{3}} & e^{\delta(u_{3})^{2}}\\
        \end{vmatrix} \\                                \label{eqn:eduut-1}
        &= & e^{\delta(3k + 2\beta)} \left(g(\mathbf{u} \mathbf{u}^{\tt T}) \frac{\delta^{3}}{2} + O(\delta^{4})\right).   \label{eqn:euut}
    \end{eqnarray}
    Here from \cref{eqn:detM-delta-middle} to (\ref{eqn:euut}) we used the Taylor expansion \cref{eqn: full_Taylor}
    on the rank one matrix $\mathbf{u} \mathbf{u}^{\tt T}$.
    
    Finally, we will use this following claim, which we shall also prove later:
    \begin{restatable}{claim}{gGood}\label{claim: gGood}
        $g(\mathbf{u}\mathbf{u}^{\tt T}) = \left((u_{2} - u_{1})(u_{3} - u_{1})(u_{3} - u_{2})\right)^{2}$.
    \end{restatable}
 
    So, if $g(\mathbf{u}\mathbf{u}^{\tt T}) = 0$, it must be the case that $u_{i} = u_{j}$ for some $i \neq j$. But in that case, we see that $\det M(e^{\delta})= 0$ for all $\delta > 0$, 
    by \cref{eqn:detM-delta-middle}, which we know to be false. Therefore, it must be the case that $g(\mathbf{u}\mathbf{u}^{\tt T}) \neq 0$. Since the leading term of $e^{\delta(3k + 2\beta)}$ is $1$, this implies that when $f''(0) = 0$, then the coefficient of $\delta^{3}$ in $\det M(e^{\delta})$ is $\frac{1}{2}g(\mathbf{u}\mathbf{u}^{\tt T}) \neq 0$. So, $\lim\limits_{\delta \rightarrow 0} \frac{\det M(e^{\delta}) }{\delta^{3}} \neq 0$.
\end{proof}

The Taylor expansion \cref{eqn: full_Taylor} and  \cref{lemma: full_determinantNonZero}
say that $\det M(e^{\delta})$ has  exact order either $\delta^2$
or $\delta^3$. We shall now finish the proof of the claims above.

\kNecessary*
\begin{proof}
    Consider the matrix $M(e^{\delta})$. Note that
    $$\begin{pmatrix}
        e^{\delta(x_{12})}\\
        e^{\delta(x_{22})}\\
        e^{\delta(\alpha x_{12} + (1 - \alpha) x_{22} + \beta)}
    \end{pmatrix} = e^{\delta(x_{12} - x_{11})} \begin{pmatrix}
        e^{\delta(x_{11})}\\
        e^{\delta(x_{12})}\\
        e^{\delta(\alpha x_{11} + (1 - \alpha) x_{12} + \beta)}
    \end{pmatrix}$$
    This implies that the second column of the matrix is a multiple of the first column, for all $\delta > 0$. Therefore, $\det M(e^{\delta})  = 0$ for all $\delta > 0$.
\end{proof}

\kWorks*
\begin{proof}
We have
$$A -kJ = N^{\tt t} \left(\begin{pmatrix}
        x_{11} & x_{12}\\
        x_{12} & x_{22}
    \end{pmatrix} - k J_2\right) N, $$
where $N = \left( \begin{smallmatrix} 
        1 & 0 & \alpha \\
        0 & 1 & 1- \alpha
    \end{smallmatrix} \right)$, and $J_2 =  \left( \begin{smallmatrix} 
        1 & 1\\
        1 & 1
    \end{smallmatrix} \right)$.
    Since $x_{11} + x_{22} \neq 2x_{12}$, $k$ is well defined
    and satisfies  
    $$(x_{11} - k)(x_{22} - k) = (x_{12} - k)^{2}.$$
    By the matrix factorization, we see that $A -kJ$ has rank at most one.
\end{proof}
 
 \gGood*
 \begin{proof}
     Note that
     \begin{align*}
        g(\mathbf{u}\mathbf{u}^{\tt T}) &= \begin{vmatrix}
            (u_{1})^{4} & (u_{1}u_{2})^{2} & (u_{1}u_{3})^{2}\\
            u_{1}u_{2} & (u_{2})^{2} & u_{2}u_{3}\\
            1 & 1 & 1\\
        \end{vmatrix} + \begin{vmatrix}
            1 & 1 & 1\\
            (u_{1}u_{2})^{2} & (u_{2})^{4} & (u_{2}u_{3})^{2}\\
            u_{1}u_{3} & u_{2}u_{3} & (u_{3})^{2}\\
        \end{vmatrix} + \begin{vmatrix}
            (u_{1})^{2} & u_{1}u_{2} & u_{1}u_{3}\\
            1 & 1 & 1\\
            (u_{1}u_{3})^{2} & (u_{2}u_{3})^{2} & (u_{3})^{4}\\
        \end{vmatrix}\\
        &\qquad + \begin{vmatrix}
            (u_{1})^{2} & u_{1}u_{2} & u_{1}u_{3}\\
            (u_{1}u_{2})^{2} & (u_{2})^{4} & (u_{2}u_{3})^{2}\\
            1 & 1 & 1\\
        \end{vmatrix} + \begin{vmatrix}
            1 & 1 & 1\\
            u_{1}u_{2} & (u_{2})^{2} & u_{2}u_{3}\\
            (u_{1}u_{3})^{2} & (u_{2}u_{3})^{2} & (u_{3})^{4}\\
        \end{vmatrix} + \begin{vmatrix}
            (u_{1})^{4} & (u_{1}u_{2})^{2} & (u_{1}u_{3})^{2}\\
            1 & 1 & 1\\
            u_{1}u_{3} & u_{2}u_{3} & (u_{3})^{2}\\
        \end{vmatrix}\\
        &= \begin{vmatrix}
            1 & 1 & 1\\
            u_{1} & u_{2} & u_{3}\\
            u_{1}^{2} & u_{2}^{2} & u_{3}^{2}\\
        \end{vmatrix} \left(-(u_{1})^{2}u_{2} - (u_{2})^{2}u_{3} - (u_{3})^{2}u_{1} + (u_{2})^{2}u_{1} + (u_{3})^{2}u_{2} + (u_{1})^{2}u_{3} \right)\\
        &= \left( (u_{2} - u_{1})(u_{3} - u_{1})(u_{3} - u_{2}) \right)^{2}
    \end{align*}
 \end{proof}
 
 Let $\lambda_i = \lambda_i(p)$ be the eigenvalues of $M(p)$, ordered by
$|\lambda_1|  \le |\lambda_2| \le |\lambda_3|$.
Clearly, for $i \in [3]$, $\lambda_{i}(p)$ are well-defined and continuous functions of $p$ (see Theorem VI.1.4 and Corollary VI.1.6 in pages 154-155 of \cite{bhatia2013matrix}).
As $M(1)$ is the all-1 matrix $J$,
$\lambda_1(1) = \lambda_2(1) = 0$ and $\lambda_3(1) =3$, 
and when $1 < p < e^{\epsilon}$,  $\lambda_{i}(p) \neq 0$ for $i \in [3]$ by \cref{cor: reductiontoMp-p-range}.
Moreover, since $M(p)$ is a positive valued matrix, the Perron
    theorem (see Theorem 8.2.8 in page 526 of \cite{horn2012matrix}) implies that $\left\lvert \lambda_{1}(p) \right\rvert \leq \left\lvert \lambda_{2}(p) \right\rvert < \left\lvert \lambda_{3}(p) \right\rvert$, and we have $\log\left(\left\lvert \lambda_{1}(p) \right\rvert/\left\lvert \lambda_{3}(p) \right\rvert\right) \neq 0$. So,
the following function $t(p)$ is well-defined on  $I_{\epsilon} = (1, e^{\epsilon})$, and is continuous as a function of $p$:
$$t(p) = 
\frac{ 
\log  \left(  \left\lvert \lambda_{2}(p) \right\rvert / \left\lvert \lambda_{3}(p) \right\rvert \right)}
{\log  \left(  \left\lvert \lambda_{1}(p) \right\rvert / \left\lvert \lambda_{3}(p) \right\rvert \right)}.$$
Clearly, $|\lambda_2| = |\lambda_1|^{t(p)} |\lambda_3|^{1-t(p)}$,
and $t(p)$ is unique satisfying this equation.

\begin{lemma}\label{lemma: full_tIrrational}
    Let $M(p)$ be a $3 \times 3$ full rank, positive real valued, symmetric matrix for $p \in I_{\epsilon}$ (where $\epsilon > 0$), with $\left(M(p)\right)_{ij} = p^{x_{ij}}$ for all $i, j \in [3]$. 
    If $t(r)$ is irrational for some $r \in I_{\epsilon}$, then $\PlEVAL(M(r))$ is $\#$P-hard.
\end{lemma}
\begin{proof}
We have
    $$\left\lvert \lambda_{2}(r) \right\rvert = \left\lvert \lambda_{1}(r) \right\rvert^{t(r)} \left\lvert \lambda_{3}(r) \right\rvert^{1 - t(r)}.$$
    If  the eigenvalues $\lambda_{i}(r)$ do not satisfy the lattice condition, then
    there are integers $n_i$ not all 0, such that  $n_1 + n_2 + n_3 =0$ and $\lambda_{1}(r)^{n_{1}}\lambda_{2}(r)^{n_{2}}\lambda_{3}(r)^{n_{3}} = 1$.
    We have $n_2 = -(n_1+n_3)\ne 0$, for otherwise $n_1 = - n_3 \ne 0$ and
    $|\lambda_{1}(r)| = |\lambda_{3}(r)|$, a contradiction.
    Then,  $$\left\lvert \lambda_{2}(r) \right\rvert = \left\lvert \lambda_{1}(r) \right\rvert^{\frac{n_{1}}{n_{1} + n_{3}}} \left\lvert \lambda_{3}(r) \right\rvert^{\frac{n_{3}}{n_{1} + n_{3}}}.$$
    By the uniqueness, $t(r) = \frac{n_{1}}{n_{1} + n_{3}}$ is rational.
     
    Therefore if $t(r)$ is irrational, then the eigenvalues $(\lambda_{1}, \lambda_{2}, \lambda_{3})$ of $M(r)$  must satisfy the lattice condition, and $\PlEVAL(M(r))$ is $\#$P-hard  by \cref{theorem: full_latticeHardnessExtension}.
\end{proof}

\begin{corollary}\label{corollary: full_tConstant}
    For  $M(p)$ given in \cref{lemma: full_tIrrational}, 
    if $t(p)$ is not a constant for all $p \in I_{\epsilon}$, then $\PlEVAL(M)$ is $\#$P-hard.
\end{corollary}
\begin{proof}
    By the intermediate value theorem for continuous functions, if $t(p)$ is not a constant function
    within $I_{\epsilon}$, then there is some $r \in I_{\epsilon}$ such that $t(r)$ is irrational. Thus $\PlEVAL(M(r))$ is $\#$P-hard. Since $\PlEVAL(M(r)) \leq \PlEVAL(M)$, this also implies that $\PlEVAL(M)$ is $\#$P-hard.
\end{proof}

\begin{remark}
    Here, it is important to note that our choice of $r$ in \cref{corollary: full_tConstant} need not be rational. In fact, it may even be the case that $r$ is transcendental.
    (See \cref{sec: full_appendixB}.) Therefore, we may not assume that $M(r)$ has rational or even algebraic values. However, as we have seen in \cref{sec: full_modelComputation}, this does not cause any issues with our model of computation, as we can continue to represent $Z_{M(r)}(G)$ as a polynomial sized tuple of integers, in terms of $\COUNT(M(r))$.
\end{remark}

\begin{lemma}\label{lemma: full_tPconvergence}
    Let  $M(p)$  be as given in \cref{lemma: full_tIrrational}. 
    Then $\PlEVAL(M)$ is $\#$P-hard, unless $t(p) = 1$ or $t(p) = \frac{1}{2}$ for all $p \in I_{\epsilon}$.
\end{lemma}
\begin{proof}
    Recall that
    $\lambda_{1}(1) = 0$, $\lambda_{2}(1) = 0$ and $\lambda_{3}(1) = 3$,
    and $\lambda_{i}(p)$ are continuous as functions of $p$.
    By definition $|\lambda_{1}(p)|/|\lambda_{2}(p)| \leq 1$.
    We claim that  either $\lim\limits_{p \rightarrow 1^{+}}\frac{|\lambda_{1}(p)|}{|\lambda_{2}(p)|} = 0$, or this ratio stays above some $c > 0$. To see that, let
    $$S(p) = \frac{\lambda_{1}\lambda_{2}}{\lambda_{3}} + \frac{\lambda_{2}\lambda_{3}}{\lambda_{1}} + \frac{\lambda_{3}\lambda_{1}}{\lambda_{2}}.$$
    $S(p)$ is a symmetric rational function. Indeed, $S(p) = \frac{s_{2}^{2} - 2s_{1}s_{3}}{s_{3}}$, 
    where
    $$s_1 = \lambda_1 + \lambda_2 + \lambda_3, ~~~s_2 = \lambda_1\lambda_2 + \lambda_2\lambda_3 + \lambda_3 \lambda_1,~~~s_3 = \lambda_1 \lambda_2 \lambda_3$$
    are the elementary symmetric polynomials of $\lambda_{i}$,
    which are polynomials in the entries of $M(p)$.
    
    Note that $S(p) = \frac{\lambda_{1}^2\lambda_{2}^2 + \lambda_{2}^2 \lambda_{3}^2 + \lambda_{3}^2\lambda_{1}^2} {\lambda_{1}\lambda_{2} \lambda_{3}}$ cannot be identically 0, for otherwise as a polynomial in $p$,  the numerator is identically 0, which would imply that 
    $\lambda_{1}\lambda_{2} = \lambda_{2} \lambda_{3} = \lambda_{3}\lambda_{1}$ are identically 0 functions
    in $p$.  But $\lambda_{3} \rightarrow 3 \ne 0$ as $p\rightarrow 1^+$, and so $\lambda_{1}, \lambda_{2}$ are identically 0. However, this contradicts $\det M(p) \ne 0$
    for $p \in I_\epsilon$.

    Expanding $S(p)$
    as a Laurent series in $(p-1)$, we have
    $S(p) = c_{k}(p-1)^{k} + c_{k + 1}(p-1)^{k + 1} + \dots$, 
    where $k\in \mathbb{Z}$, and $c_{k} \neq 0$.
    Since ${\lambda_{1}\lambda_{2}}/{\lambda_{3}} \rightarrow 0$
    as $p \rightarrow 1^+$,
    and $|{\lambda_{3}}| \left|\frac{\lambda_{1}}{\lambda_{2}}\right|$ stays bounded above,
    it follows that $ k <0$ iff $S(p) \rightarrow \infty$ iff $\lambda_{1}/\lambda_{2} \rightarrow 0$, and $k \ge 0$ iff 
    $|\lambda_{1}/\lambda_{2}| \ge c > 0$, for all $p \in I_{\epsilon}$.

    Suppose $|\lambda_{1}/\lambda_{2}| \ge c > 0$ for all $p \in I_{\epsilon}$,
    then 
    $$t(p) =  1 -
    \frac{ 
    \log  \left(  \left\lvert \lambda_{1}(p) \right\rvert / \left\lvert \lambda_{2}(p) \right\rvert \right)}
    {\log  \left(  \left\lvert \lambda_{1}(p) \right\rvert / \left\lvert \lambda_{3}(p) \right\rvert \right)} \rightarrow 1.$$
    Therefore the only possibility that $t(p)$ is a constant on  $I_{\epsilon}$
    is that it is constant 1. If $t(p)$ is not  constant $1$, 
    then $\PlEVAL(M)$ is $\#$P-hard by \cref{corollary: full_tConstant}.
    
    Now suppose
    $\lim\limits_{p \rightarrow 1^{+}}\frac{|\lambda_{1}(p)|}{|\lambda_{2}(p)|} = 0$.
    As we already noted, $s_2$ is a polynomial in the entries of $M(p)$,
    which we can express as a polynomial in $(p-1)$. It is 
    not identically 0.
    This can be seen by $s_2 = \lambda_{2} ( \lambda_{1} + \lambda_{3} ( 1 + 
    \frac{ \lambda_{1}}{\lambda_{2}}) )$, where
    the second factor has the same limit as $\lambda_{3} \rightarrow 3$. So if $s_2$ were identically 0
    we would have $\lambda_{2}$ identically 0,
    contradicting $\det M(p) \ne 0$
    for $p \in I_\epsilon$.
    Let $s_2 = a_{j} (p-1)^{j} + a_{j+1}(p-1)^{j+1} + \ldots$ be in increasing power terms, 
    where $a_{j} \ne 0$ is the first nonzero term.
    It has the same order as $\lambda_{2}$ when $p \rightarrow 1^+$,
    since $s_2 = \lambda_{2} ( \lambda_{1} + \lambda_{3} ( 1 + 
    \frac{ \lambda_{1}}{\lambda_{2}}) )$, where $\lambda_{1}\rightarrow 0$, $\frac{ \lambda_{1}}{\lambda_{2}} \rightarrow 0$
    and $\lambda_{3} \rightarrow 3$.
    So $\lambda_{2}$ also has the exact order $j$.
   The important point is  that   this exact order $j$ is an integer.
    
    From \cref{lemma: full_determinantNonZero}, we know that
    $\lambda_{1}\lambda_{2}\lambda_{3}$ is either
    of exact order $\Theta((p-1)^{2})$ or $\Theta((p-1)^{3})$.
    Since $\lambda_{3} \rightarrow 3$ the same is true for
    the product 
    $\lambda_{1}\lambda_{2}$.
    As both $\lambda_{1}, \lambda_{2}  \rightarrow 0$, and we are
    in the case $|\lambda_{1}|/|\lambda_{2}| \rightarrow 0$,
    the only possibility is $j=1$, i.e.,  $\lambda_{1} = \Theta((p-1)^2)$
    and $\lambda_{2} = \Theta((p-1)^1)$.
    
    In this case, we see that
    $\lim\limits_{p \rightarrow 1^+}t(p) = \frac{1}{2}$.
    Therefore, once again by \cref{corollary: full_tConstant}, if $t(p)$ is not a constant  $\frac{1}{2}$ for $p \in I_{\epsilon}$, then $\PlEVAL(M)$ is $\#$P-hard.
\end{proof}

\begin{lemma}\label{lemma: full_tOneGood}
Let  $M(p)$  be given in \cref{lemma: full_tIrrational}.    If $t(p) = 1$ for all $p \in I_{\epsilon}$, then $\PlEVAL(M)$ is $\#$P-hard.
\end{lemma}
\begin{proof}
    From $t(p) = 1$ we have $\left\lvert \lambda_{2}(p) \right\rvert = \left\lvert \lambda_{1}(p) \right\rvert$. Consider the matrix $M(p)^{2}$. Since $M(p)$ is a positive matrix, so is $M(p)^{2}$. Moreover, its eigenvalues are $\lambda_{1}(p)^{2} = \lambda_{2}(p)^{2} < \lambda_{3}(p)^{2}$. Therefore, $\PlEVAL(M(p)^{2})$ is $\#$P-hard, as a consequence of \cref{theorem: full_latticeHardnessExtension}. Since $\PlEVAL(M(p)^{2}) \leq \PlEVAL(M(p)) \leq \PlEVAL(M)$, it also means that $\PlEVAL(M)$ is $\#$P-hard.
\end{proof}

\begin{lemma}\label{lemma: full_tHalfImpossible}
    Let  $M(p)$  be given in \cref{lemma: full_tIrrational}.  
    Then  $t(p)$
    is not a constant $\frac{1}{2}$ on the interval $I_{\epsilon}$.
\end{lemma}
\begin{proof}
    Let $\mu_i = \lambda_i^2 = (\lambda_i(p))^2$,
    we define the function
    $$ F(p) = (\mu_1^2 - \mu_2 \mu_3) (\mu_2^2 - \mu_3 \mu_1) (\mu_3^2 - \mu_1 \mu_2).$$
    Note that $F(p)$ is a symmetric polynomial of the eigenvalues $\mu_i$ of
    $M(p)^2$, and in fact 
    \[F(p)= s_{1}^{3}s_{3} - s_{2}^{3},\]
    where $s_i$ are the elementary symmetric polynomials of $\{\mu_1, \mu_2, \mu_3\}$,
    \[s_1 = \mu_1 + \mu_2 + \mu_3, ~~~s_2 = \mu_1\mu_2 + \mu_2\mu_3 + \mu_3 \mu_1,~~~s_3 = \mu_1 \mu_2 \mu_3.\]
    As coefficients of $\det(xI - M(p)^2)$,
    they are polynomials in the entries of $M(p)$.
    
    Suppose for a contradiction that  $t(p) = \frac{1}{2}$ on
    the interval $I_{\epsilon}$, then $\mu_{2}^2 = \mu_{1}\mu_{3}$ in that interval. Then $F(p) = 0$ for all $p \in I_{\epsilon}$. Since $F(p)$ is a polynomial in $p$,
    we have a polynomial identity
    $$s_{1}^{3}s_{3} = s_{2}^{3}.$$
    As $p \rightarrow 1^+$, both $s_1$ and $s_3$ are nonzero, and so
    all three are nonzero polynomials.
    By the unique factorization of polynomials $s_{1} \mid s_{2}$,
    and so $s_{2} = s_{1}f$ for some nonzero polynomial $f(p)$.
    It follows that $f^3 = s_3 = \det (M(p))^{2}$.
    The exact order of any irreducible polynomial $q$
    in $f^3$ is ${\rm ord}_q(f^3) = 3 \cdot \text{ord}_q(f) \equiv 0 \bmod 3$,
    which is also $2 \cdot {\rm ord}_q(\det (M(p)))$. Thus 
    $\det(M(p))$ is a cubic power of a polynomial, $\det(M(p)) = g^{3}$.

    Now
    \[
        g(p)^3
        = \det(M(p))
        = p^{x_{11} + x_{22} + x_{33}} + p^{x_{12} + x_{23} + x_{13}} + p^{x_{12} + x_{23} + x_{13}} - p^{x_{11} + 2x_{23}} - p^{x_{33} + 2x_{12}} - p^{x_{22} + 2x_{13}}.\]
    In the expression for
    $\det(M(p))$ there are three positive and three negative terms. If any cancellation occurs, an equal number of positive and negative terms are
    cancelled, hence the nonzero polynomial  $\det(M(p))$ 
    has either 2 or 4 or 6 terms  with equal number having $+ 1$ and $-1$ coefficients
    after cancellation. So $g(p)$ is not a monomial.
    We may write  $g(p) = c_{1}p^{x_{1}} + \dots + c_{k}p^{x_{k}}$ with $x_{1} > \dots > x_{k}$ and nonzero integers $c_{1}, \dots, c_{k}$, with $k \ge 2$.
    Then  $g(p)^{3}$ has the following terms 
    which have distinct degrees and cannot be cancelled:
        $$c_{1}^{3}p^{3x_{1}},  ~~~3c_{1}^{2}c_{2}p^{2x_{1} + x_{2}}, ~~~3c_{k}^{2}c_{k - 1}p^{2x_{k} + x_{k - 1}}, ~~~c_{k}^{3}p^{3x_{k}}.$$
    In terms of  monomial terms with $\pm 1$ coefficients there are at least
    8 terms. These cannot be matched by at most 6 monomial terms with $\pm 1$ coefficients.
    This contradiction proves the lemma.
\end{proof}

\begin{theorem}\label{theorem: full_positiveHard}
    If $M$ is a $3 \times 3$ full rank, positive real valued, symmetric matrix, then $\PlEVAL(M)$ is $\#$P-hard.
\end{theorem}
\begin{proof}
    We define $M(p)$ for $p \in I_{\epsilon}$ for some $\epsilon > 0$, as 
    was done after \cref{cor: reductiontoMp-p-range}. From \cref{lemma: full_tHalfImpossible}, $t(p)$ is not  a constant $\frac{1}{2}$ for all $p \in I_{\epsilon}$. On the other hand, if $t(p) = 1$ for all $p \in I_{\epsilon}$,  from \cref{lemma: full_tOneGood}, we know that $\PlEVAL(M)$ is $\#$P-hard. Finally, if $t(p)$ is not constant  $1$ or $\frac{1}{2}$
    for all $p \in I_{\epsilon}$, then  $\PlEVAL(M)$ is $\#$P-hard from \cref{lemma: full_tPconvergence}.
\end{proof}
\section{Dichotomy for \texorpdfstring{$3 \times 3$}{3 x 3} matrices}\label{sec: full_finalDichotomy}

From \cref{sec: full_positiveHardness}, if  $M$ is $3 \times 3$ full rank, positive real valued, symmetric matrix, then $\PlEVAL(M)$ is $\#$P-hard. We will now use \cref{theorem: full_positiveHard} to establish a dichotomy for all $3 \times 3$ matrices.

\begin{lemma}\label{lemma: full_nonNegativeHardness}
    If $M$ is a $3 \times 3$ full rank, non-negative real valued, symmetric matrix that is irreducible, then $\PlEVAL(M)$ is $\#$P-hard.
\end{lemma}
\begin{proof}
    As we saw in \cref{sec: full_twinned}, a $3 \times 3$ matrix can only be bipartite if it is a twinned matrix. Since $M$ is full rank, it is therefore non-bipartite.
    The proof of this lemma is then the same as that of \cref{theorem: full_qTimesqNonBipLatticeNonNegative}, and is a consequence of \cref{theorem: full_positiveHard}.
\end{proof}

\begin{lemma}\label{lemma: full_rankTwoMixed}
	If $M$ is a $3 \times 3$ rank two, real valued, symmetric matrix that is irreducible and not twinned, then $\PlEVAL(M)$ is $\#$P-hard.
\end{lemma}
\begin{proof}
    If $M$ is a rank two matrix that is not twinned, it must be of the form
    $$M = N^{\tt T} \begin{pmatrix}
        x & y\\
        y & z
    \end{pmatrix} N, ~~~ \mbox{where}~~~N = \begin{pmatrix} 
        1 & 0 & \alpha \\
        0 & 1 & \beta
    \end{pmatrix},$$
	for some $\alpha \neq 0$ and $\beta \neq 0$.
	
    If $xz = y^{2}$, 
    then $M$ has rank at most one,  contrary to assumption.
    So, $xz \neq y^{2}$.
    Now, we consider $T_{2}M$. We have
	$$T_{2}M = \begin{pmatrix}
		x^{2} & y^{2} & \alpha^{2} x^{2} + \beta^{2} y^{2} + 2\alpha\beta xy\\
		y^{2} & z^{2} & \alpha^{2} y^{2} + \beta^{2} z^{2} + 2\alpha\beta yz\\
		\alpha^{2} x^{2} + \beta^{2} y^{2} + 2\alpha\beta xy & \alpha^{2} y^{2} + \beta^{2} z^{2} + 2\alpha\beta yz & \left(\alpha^{2}x + 2\alpha\beta y + \beta^{2}z\right)^{2}\\
	\end{pmatrix}.$$
    Miraculously,
    $$\det(T_{2}M) = 2\alpha^{2}\beta^{2}\left(xz - y^{2}\right)^{3}. $$
    Thus, $T_{2}M$ is a full rank matrix. Since $M$ is irreducible, it also follows that $T_{2}M$ is irreducible. Moreover, since $T_{2}M$ is non-negative valued, we see from \cref{lemma: full_nonNegativeHardness} that $\PlEVAL(T_{2}M)$ is $\#$P-hard. The lemma follows
    from  $\PlEVAL(T_{2}M) \leq \PlEVAL(M)$.
\end{proof}

\begin{lemma}\label{lemma: full_finalHardness}
	If $M$ is a $3 \times 3$ full rank, real valued, symmetric matrix that is irreducible, then $\PlEVAL(M)$ is $\#$P-hard.
\end{lemma}
\begin{proof}
    We consider the matrix $T_{2}M$. We note that this is a non-negative valued matrix. Since $M$ is irreducible, neither is $T_{2}M$. If $T_{2}M$ is a full rank matrix, the hardness of $\PlEVAL(T_{2}M)$, and therefore, the hardness of $\PlEVAL(M)$ follows from \cref{lemma: full_nonNegativeHardness}.
    
     Next, suppose $T_{2}M$ has rank two. If $T_{2}M$ is not twinned, then the hardness of $\PlEVAL(T_{2}M)$, and therefore, the hardness of $\PlEVAL(M)$ follows from \cref{lemma: full_rankTwoMixed}. Therefore, let us consider the case where $T_{2}M$ is twinned. From \cref{theorem: full_rankTwoMultipleC}, we can see that $\PlEVAL(T_{2}M)$ can be polynomially tractable only if $T_{2}M$ is rank one, or if $T_{2}M$ is reducible, or if $T_{2}M$ is of the form
    $$T_{2}M = \begin{pmatrix}
        0 & 0 & y\\
        0 & 0 & z\\
        y & z & 0\\
    \end{pmatrix}$$
    upto permutations, in which case, it must be the case that $M$ is also of the form above, and must be a rank two matrix. Since none of these are true, it follows from \cref{theorem: full_rankTwoMultipleC} that $\PlEVAL(T_{2}M)$ is $\#$P-hard, and therefore, so is $\PlEVAL(M)$.
    
    Finally, suppose $T_{2}M$ has rank one.
    If $T_{2}M$ has any zero entries, it can be easily checked  that $T_{2}M$, and therefore $M$ is reducible. So,  $T_{2}M$ has no zero entries. In this case, we consider $\mathcal{M}(\mathbf{1})$. The entries of this matrix may be denoted by $\sigma_{ij} \in \{+1, -1\}$. Since $T_{2}M$ is a rank one matrix, all three rows are multiples of each other.
    If we further have that two rows of $\mathcal{M}(\mathbf{1})$ are multiples of each other, then  $M$ itself is not full rank. Therefore, no two rows of $\mathcal{M}(\mathbf{1})$ are multiples of each other. Now, we consider $S_{2}\mathcal{M}(\mathbf{1}) = \mathcal{M}(\mathbf{1})^{2}$. Clearly,
    $\left(\mathcal{M}(\mathbf{1})^{2}\right)_{ii} = \sigma_{i1}^{2} + \sigma_{i2}^{2} + \sigma_{i3}^{2} = 3$ for $i \in 3$. Moreover, if $i \neq j$, since no two rows are multiples of each other, we see that $\left(\mathcal{M}(\mathbf{1})^{2}\right)_{ij} = \sigma_{i1}\sigma_{j1} + \sigma_{i2}\sigma_{j2} + \sigma_{i3}\sigma_{j3} \in \{+1, -1\}$,
    as exactly one pair cancels. So, the set $\{3\}$ is
    a generating set for the entries of $\mathcal{M}(\mathbf{1})^{2}$, and the hardness of $\PlEVAL(\mathcal{M}(\mathbf{1})^{2})$ follows from \cref{lemma: full_thickeningBasic}. Since $\PlEVAL(\mathcal{M}(\mathbf{1})^{2}) \leq \PlEVAL(\mathcal{M}(\mathbf{1})) \leq \PlEVAL(M)$, we see that $\PlEVAL(M)$ is also $\#$P-hard.
\end{proof}

Combining all our results, we have our final dichotomy:

\begin{theorem}\label{theorem: full_mainTheorem}
	If $M$ is a $3 \times 3$ real valued, symmetric matrix, then $\PlEVAL(M)$ is $\#$P-hard, unless $M$ is of one of the following forms upto a permutation
	of rows and columns and in which case, $\PlEVAL(M)$ is tractable in polynomial time:
	\begin{enumerate}
		\item
		$$M = \begin{pmatrix}
			x^{2} & xy & xz\\
			xy & y^{2} & yz\\
			xz & yz & z^{2}\\
		\end{pmatrix} = {\bf u} {\bf u}^{\tt t}.$$
		
		\item 
		$$M = \begin{pmatrix}
			x & y & 0\\
			y & z & 0\\
			0 & 0 & t\\
		\end{pmatrix}$$
		such that
    		\[   \mbox{(1)~~} xz = y^2,~~~~
        	\mbox{(2)~~} y = 0,~~~~
        	\mbox{(3)~~} x = z, ~~~~\mbox{or}~~~~
        	\mbox{(4)~~} xz = -y^2 ~~\&~~ x = -z. \]
	
		\item 
		$$M = \begin{pmatrix}
			0 & 0 & x\\
			0 & 0 & y\\
			x & y & 0\\
		\end{pmatrix}.$$
	\end{enumerate}
\end{theorem}
\begin{proof}
    The listed forms are all polynomial-time tractable.
    If $M$ has rank $\le 1$, then $M$ has form $1$. Now we assume
    $M$ has rank $\ge 2$.
    If $M$ is reducible, by \cref{lemma: full_trivialCase}
    and \cref{theorem: full_domain2Hardness}, $\PlEVAL(M)$ is $\#$P-hard unless $M$ has form $2$. Below, we assume $M$ is irreducible.
    If $M$ is a rank two twinned matrix, by \cref{theorem: full_rankTwoMultipleC}  $\PlEVAL(M)$ is $\#$P-hard unless it has form  $3$ (i.e. $M$ is bipartite). If $M$ has rank two and not twinned, by \cref{lemma: full_rankTwoMixed}, $\PlEVAL(M)$ is $\#$P-hard. So now we assume $M$ is a full rank matrix. Then, by \cref{lemma: full_finalHardness}  $\PlEVAL(M)$ is $\#$P-hard.
\end{proof}

\printbibliography

\appendix

\section{Appendix A}\label{sec: full_appendixA}

We will  briefly sketch a proof of \cref{theorem: full_degreeBoundedHardness}. 
The main ingredients of the theorem are the 
edge gadget $\mathcal{P}_{n, p}$, and the 
vertex gadget $\mathcal{R}_{d, n, p}$. 
The edge gadget $\mathcal{P}_{n, p}$ is just 
$S_{2}T_{p}S_{n}(e)$, where $e$ is an edge. 
The vertex gadget $\mathcal{R}_{d, n, p}$ is constructed 
by replacing each edge on a simple cycle with $d$ vertices 
with the edge gadget $\mathcal{P}_{n, p}$, 
and adding a dangling edge to each of the $d$ vertices 
$F_{1}, \dots, F_{d}$ on the original $d$-cycle.

\begin{figure}[ht]
	\centering
	\scalebox{0.8}{\input{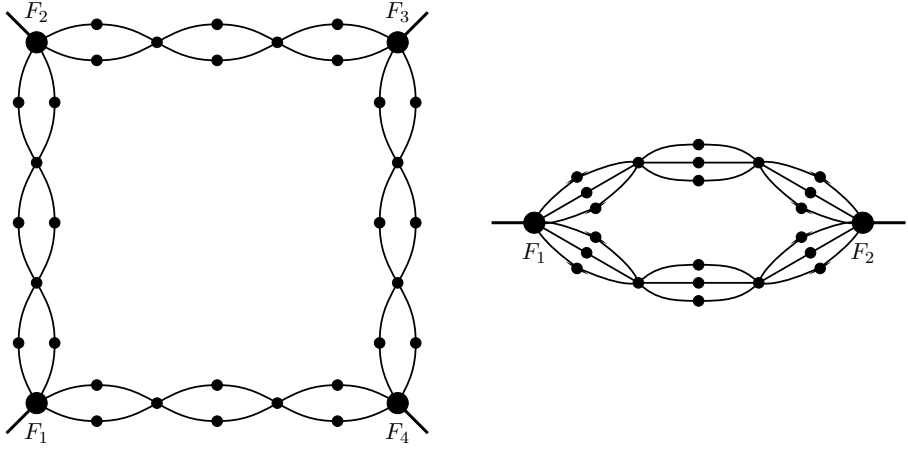}}
	\caption{The gadgets $\mathcal{R}_{4, 3, 2}$ and $\mathcal{R}_{2, 3, 3}$}
	\label{fig:degreeGadget}
\end{figure}

Now, consider a planar graph $G = (V, E)$. 
For any $u \in V$, consider the edges incident on $u$. 
Since the graph $G$ is planar, 
we we may arbitrarily number one of these edges as $1$, 
and then go clockwise through all the other edges incident on $u$
starting from the edge labelled $1$, 
and label these edges $2, 3, \dots, \deg(u)$. 
So, the labelling function 
$\ell_{u}: \{\{u, v\} \in E: v \in V\} \rightarrow [\deg(u)]$ 
is well defined. 
Now, we construct a graph $G_{n, p}$ as follows: 
We replace each vertex $u \in V$ 
with the gadget $\mathcal{R}_{\deg(u), n, p}$ 
(which we can label as $\mathcal{R}_{\deg(u), n, p}(u)$). 
We know that there are $d$ dangling edges on this gadget, 
incident on each of the vertices $F_{1}(u), \dots, F_{\deg(u)}(u)$. 
Now, we consider each edge $\{u, v\} \in E$, 
and for each such edge, we connect together, 
the dangling edges in the graph $G_{n, p}$, 
between the vertices $F_{\ell_{u}(\{u, v\})}(u)$ and $F_{\ell_{v}(\{u, v\})}(v)$.
We note that by this construction, 
the graph $G_{n, p}$ is planar, for all $n, p$.
We will construct $G_{n, p}$ for various choices of $n, p$, 
and then interpolate to prove \cref{theorem: full_degreeBoundedHardness}.

First, an appropriate $p$ is chosen such that 
$B = T_{p}(M'D^{[2]}M')$ is non-degenerate. 
Such a $p$ exists as a result of the following technical lemma, 
whose proof can be found in \cite{govorov2020dichotomy}.

\begin{lemma}\label{lemma: full_govorovTechnical}
     Let $A$ and $D$ be $n \times n$ matrices, 
     where $A$ is real symmetric with all columns nonzero 
     and pairwise linearly independent, 
     and $D$ is positive diagonal. 
     Then for all sufficiently large positive integers $p$, 
     the matrix $B = T_{p}(ADA)$ is non-degenerate.
\end{lemma}

After an appropriate $p$ is chosen, 
we note that the gadget $\mathcal{P}_{n, p}$ effectively simulates edge weights given by the matrix
$$L^{(n)} = BD^{[2p]}B \cdots BD^{[2p]}B = B(D^{[2p]}B)^{n - 1} = \big(D^{[2p]}\big)^{-\frac{1}{2}} \Big(\big(D^{[2p]}\big)^{\frac{1}{2}}B\big(D^{[2p]}\big)^{\frac{1}{2}} \Big)^{n}\big(D^{[2p]}\big)^{-\frac{1}{2}}$$

Since $\tilde{B} = \left(D^{[2p]}\right)^{\frac{1}{2}}B\left(D^{[2p]}\right)^{\frac{1}{2}}$ 
is a real symmetric matrix, we know that it can be diagonalized.
Therefore, there exists a real valued, orthonormal symmetric matrix $S$,
and a real valued, diagonal matrix $J$, 
such that $\tilde{B} = SJS^{T}$. So, we see that
$$L^{(n)} = \left(D^{[2p]}\right)^{-\frac{1}{2}}\tilde{B}^{n}\left(D^{[2p]}\right)^{-\frac{1}{2}} = \left(D^{[2p]}\right)^{-\frac{1}{2}}SJ^{n}S^{T}\left(D^{[2p]}\right)^{-\frac{1}{2}}$$
From \cref{lemma: full_govorovTechnical}, 
we know that $B$ is non-degenerate, and therefore, so is $\tilde{B}$.
Therefore, $J = \text{diag}(\lambda_{1}, \lambda_{2})$, 
where $\lambda_{i} \neq 0$ for all $i \in [2]$. 
So, we see that for any $i, j \in [2]$, 
${L^{(n)}}_{ij} = a_{ij1}\lambda_{1}^{n} + a_{ij2}\lambda_{2}^{n}$,
where $a_{ij1}, a_{ij2}$ depend on $S$ and $D^{[2p]}$, 
but not on $J$ or $n$.

Now, we finally consider $Z_{M', \mathcal{D}}(G_{n, p})$. 
For convenience, 
let $V' = \bigcup_{u \in V}\{F_{1}(u), \dots, F_{\deg(u)}(u)\}$ 
denote the subset of the vertices in $G_{n, p}$ 
that lie on the original $d$ cycle of $\mathcal{R}_{d, n, p}(u)$, 
and let $E' = \big\{\{F_{\ell_{u}(\{u, v\})}(u), F_{\ell_{v}(\{u, v\})}(v)\}: \{u, v\} \in E \big\}$ 
be a subset of the edges in $G_{n, p}$ that were created 
by connecting the dangling edges of the vertex gadgets. 
Then, we see that
$$Z_{M', \mathcal{D}}(G_{n, p}) = \sum_{\sigma: V' \rightarrow [2]}\left(\prod_{\{u, v\} \in E'}M'_{\sigma(u), \sigma(v)}\right)\left(\prod_{u \in V}\prod_{i = 1}^{\deg(u)}D^{[2p + 1]}_{\sigma(F_{i}(u))} \cdot {L^{(n)}}_{\sigma(F_{i}(u)), \sigma(F_{i + 1}(u))}\right)$$
After rearranging the terms, 
we see that
$$Z_{M', \mathcal{D}}(G_{n, p}) = \sum_{x_{1}, x_{2}: x_{1} + x_{2} = 2|E|}c_{(x_{1}, x_{2})}\left(\lambda_{1}^{x_{1}}\lambda_{2}^{x_{2}}\right)^{n}$$
where the constants $c_{(x_{1}, x_{2})}$ 
do not depend on $\lambda_{i}$ or $n$. 
By computing $Z_{M', \mathcal{D}}(G_{n, p})$ for $n \in [2|E|]$, 
we obtain a Vandermonde system of linear equations. 
If there exist any $(x_{1}, x_{2}) \neq (y_{1}, y_{2})$ such that $\lambda_{1}^{x_{1}}\lambda_{2}^{x_{2}} = \lambda_{1}^{y_{1}}\lambda_{2}^{y_{2}}$, 
then we can remove the column of the matrix corresponding to
$\left(\lambda_{1}^{y_{1}}\lambda_{2}^{y_{2}}\right)^{n}$, 
and combine the two coefficients 
$c_{(x_{1}, x_{2})}$ and $c_{(y_{1}, y_{2})}$ 
into a new coefficient $c_{(x_{1}, x_{2}), (y_{1}, y_{2})}$, 
which represents the sum of the two previous coefficients. 
By repeating this process, we obtain a 
full rank Vandermonde system of linear equations, 
which can be solved. 
After solving this system of equations, we can compute
\begin{align*}
    X
    &= \sum_{x_{1}, x_{2}: x_{1} + x_{2} = 2|E|}c_{(x_{1}, x_{2})}\\
    &= \sum_{x_{1}, x_{2}: x_{1} + x_{2} = 2|E|}c_{(x_{1}, x_{2})}(1)^{x_{1}}(1)^{x_{2}}\\
    &= \sum_{\sigma: V' \rightarrow [2]}\left(\prod_{\{u, v\} \in E'}M'_{\sigma(u), \sigma(v)}\right)\left(\prod_{u \in V}\prod_{i = 1}^{\deg(u)}D^{[2p + 1]}_{\sigma(F_{i}(u))}L^{(0)}_{\sigma(F_{i}(u)), \sigma(F_{i + 1}(u))}\right)
\end{align*}
where
\begin{align*}
    L^{(0)}
    &= \left(D^{[2p]}\right)^{-\frac{1}{2}}SIS^{T}\left(D^{[2p]}\right)^{-\frac{1}{2}}\\
    &= \left(D^{[2p]}\right)^{-\frac{1}{2}}(I)\left(D^{[2p]}\right)^{-\frac{1}{2}}\\
    &= \left(D^{[2p]}\right)^{-1}
\end{align*}

Since $L^{(0)}$ is a diagonal matrix, 
any permutation $\sigma: V' \rightarrow [2]$ 
contributes to the sum $X$ 
only if $\sigma(F_{i}(u)) = \sigma(F_{i + 1}(u))$ 
for all $i \in [\deg(u)]$, for all $u \in V$. 
In that case, we see that
$$X = \sum_{\sigma': V \rightarrow [2]}\left(\prod_{\{u, v\} \in E}M'_{\sigma'(u), \sigma'(v)}\right)\left(\prod_{u \in V}\left(\frac{D^{[2p + 1]}_{\sigma'(u)}}{D^{[2p]}_{\sigma'(u)}}\right)^{\deg(u)}\right)$$

In effect, we can compute $X = Z_{M', \hat{\mathcal{D}}}(G)$, 
where $\hat{\mathcal{D}} = \{\hat{D}^{[r]}\}_{r \in \mathbb{N}}$, 
such that
$$\hat{D}^{[r]} = \begin{pmatrix}
    \left(\frac{1 + c^{2p + 1}}{1 + c^{2p}}\right)^{r} & 0\\
    0 & 1\\
\end{pmatrix}$$

In other words, we see that 
$\hat{\mathcal{D}}^{[r]} = \big(\hat{D}\big)^{r}$, 
where 
$$\hat{D} = \begin{pmatrix}
    \left(\frac{1 + c^{2p + 1}}{1 + c^{2p}}\right) & 0\\
    0 & 1\\
\end{pmatrix}$$

So, we see that we can compute
\begin{align*}
    X
    &= \sum_{\sigma': V \rightarrow [2]}\left(\prod_{\{u, v\} \in E}M'_{\sigma'(u), \sigma'(v)}\right)\left(\prod_{u \in V}\left(\hat{D}_{\sigma'(u)}\right)^{\deg(u)} \right)\\
    &= \sum_{\sigma': V \rightarrow [2]}\left(\prod_{\{u, v\} \in E}M'_{\sigma'(u), \sigma'(v)} \left(\hat{D}_{\sigma'(u)}\right)\left(\hat{D}_{\sigma'(v)}\right)\right)\\
    &= Z_{N_{p}}(G)
\end{align*}
where
$$N_{p} = \begin{pmatrix}
    \left(\frac{1 + c^{2p + 1}}{1 + c^{2p}}\right)^{2}M'_{11} & \left(\frac{1 + c^{2p + 1}}{1 + c^{2p}}\right) M'_{12}\\
    \left(\frac{1 + c^{2p + 1}}{1 + c^{2p}}\right)M'_{12} & M'_{22}
\end{pmatrix}$$ 
as in \cref{theorem: full_degreeBoundedHardness}.

$$\therefore \PlEVAL(N_{p}) \equiv \PlEVAL(M', \hat{\mathcal{D}}) \leq \PlEVAL(M', \mathcal{D})$$

This completes the proof of \cref{theorem: full_degreeBoundedHardness}.

\section{Appendix B}\label{sec: full_appendixB}

\begin{theorem}\label{thm:mapping-alg-to-rational}
There is a monotonic increasing
continuous function $F: (0,1)
\rightarrow (0,1)$ such that 
$F(a)$ is rational iff $a \in (0, 1)$ is an 
algebraic number.
\end{theorem}
\begin{proof}
Consider  $X = 
 \overline{\mathbb{Q}} \cap (0,1)$
and $Y = \mathbb{Q}
 \cap (0,1)$, where
$\mathbb{Q}$ and $\overline{\mathbb{Q}}$
are  the rational
and algebraic numbers, respectively.
By Cantor's density theorem (see Theorem 4.9 in page 83 of ~\cite{hrbacek1999introduction}),
any two countable unbounded dense orders
are order isomorphic.
This means that there is a
1-1 onto
map $f$ from $X$ to $Y$ that is order preserving.

Now we can extend $f$ to be a
continuous function
$F: (0,1) \rightarrow (0,1)$ as follows:
For any $r \in (0,1)$,
let $X_r = \{x \in X : x < r\}$
and $X^r =  \{x \in X : x > r\}$.
Both $X_r$ and $X^r$ are bounded
subsets of $X$.
So the following $\sup$ and $\inf$  both exist and are finite.
We claim they are equal:
\begin{equation}\label{eqn:inf-sup}
\sup_{x \in X_r} f(x)
= \inf_{x \in  X^r} f(x).
\end{equation}
Since $f$ preserves order,
clearly $\sup_{x \in X_r} f(x)
\le \inf_{x \in  X^r} f(x)$.
Suppose they are unequal. Then
by $Y$ being dense,
there exist $y, y' \in Y$ such that
$\sup_{x \in X_r} f(x) < y < y' <
\inf_{x \in  X^r} f(x)$.
Consider $a = f^{-1}(y) \in X$.
If $a < r$, then there exists some
$x \in X_r$, such that $x >a$.
Then $f(x) > f(a) =y > \sup_{x \in X_r} f(x)$, a contradiction.
Hence $a \ge r$.
Similarly, $a' = f^{-1}(y') \in X$, and $a' \le r$.
But then $y' = f(a') \le f(a) =y$,
a contradiction.

Now we define $F(r)$
to be the common value in
\cref{eqn:inf-sup}.
We claim that this $F$ is an extension of the order
isomorphism $f$.
Let $r \in X$. For any
$\epsilon > 0$, there exists
$y \in Y$ such that
$f(r) < y < f(r) + \epsilon$.
Then $f^{-1}(y) \in X^r$ 
and $F(r) = \inf_{x \in  X^r} f(x) \le y < f(r) + \epsilon$.
Similarly, $F(r) > f(r) - \epsilon$.
Since $\epsilon > 0$
is arbitrary,   $F(r) = f(r)$.

Now we show that $F$ is continuous.
For any $r \in (0, 1)$ and $\epsilon >0$,
there exist $x \in X_r$
and $x' \in X^r$, such that
$F(r) -\epsilon < f(x) <
F(r) < f(x') < F(r) +\epsilon $.
We claim that for every
$z \in (x, x')$,
\[|F(z) - F(r)| < \epsilon.\]
This is because by the definition of inf and sup,
$f(x) \le F(z) \le f(x')$.
Hence
\[F(r) -\epsilon < f(x) \le F(z) 
\le  f(x') < F(r) +\epsilon.
\]
Finally, it is easy to see
that $F$ is strictly increasing.
Suppose $x < x'$. There are
$a, a' \in X$, such that
$x < a < a' < x'$.
Then $F(x) \le f(a) < f(a') \le F(x')$.
We conclude that if $F(x)$
is rational, then $x$ is
algebraic. Indeed,
$a = f^{-1}(F(x))$
is defined and is algebraic,
and $F(a) = f(a) = F(x)$.
Thus $x=a$ is algebraic.
Conversely, if $x$ is algebraic, as $F$ is an extension of $f$,
$F(x) = f(x)$, which is algebraic by the given property of $f$.
\end{proof}

\end{document}